\documentclass[11pt]{article}
\usepackage[margin=0.8in]{geometry}
\usepackage{authblk}

\usepackage[utf8]{inputenc} 
\usepackage[T1]{fontenc}    
\usepackage[hidelinks]{hyperref}       
\usepackage{url}            
\usepackage{booktabs}       
\usepackage{amsfonts}       
\usepackage{nicefrac}       
\usepackage{microtype}      
\usepackage{xcolor}         

 \newcommand\numberthis{\addtocounter{equation}{1}\tag{\theequation}}
\usepackage{amsmath}
\DeclareMathOperator*{\argmax}{arg\,max}
\usepackage{amsthm,thmtools}
\usepackage{amssymb}
\usepackage{mathtools}
\usepackage[noend,ruled,linesnumbered]{algorithm2e}
\usepackage{soul}
\renewcommand{\algorithmcfname}{Mechanism}
\usepackage{multicol}
\usepackage[normalem]{ulem}
\usepackage{cleveref}

\begin{document}

\title{Getting More by Knowing Less:\\ Bayesian Incentive Compatible Mechanisms for Fair Division}

\author[a]{Vasilis Gkatzelis}
\author[b]{Alexandros Psomas}
\author[a]{Xizhi Tan}
\author[b]{Paritosh Verma}

\affil[a]{Department of Computer Science, Drexel University.

\texttt{{\{gkatz,xizhi\}@drexel.edu}}}

\affil[b]{Department of Computer Science, Purdue University.

\texttt{{\{apsomas,verma136\}@cs.purdue.edu}}}

\date{}

\newtheorem{example}{Example}
\newtheorem{theorem}{Theorem}
\newtheorem{lemma}{Lemma}
\newtheorem{corollary}{Corollary}
\newtheorem{definition}{Definition}
\newtheorem{proposition}{Proposition}
\newtheorem{claim}{Claim}

\newcommand{\sdpoplus}{\mathrm{SD}^+\text{efficiency}}
\newcommand{\IH}{\mathrm{IH}}
\newcommand{\sdpo}{\mathrm{SD}\text{ efficiency}}
\newcommand{\agents}{\mathcal{N}}
\newcommand{\items}{\mathcal{M}}
\newcommand{\allo}{\mathbf{x}}
\newcommand{\alloy}{\mathbf{y}}
\newcommand{\alloa}{\mathbf{a}}
\newcommand{\allob}{\mathbf{b}}
\newcommand{\alloi}{\mathbf{x}}
\newcommand{\val}{\mathbf{v}}
\newcommand{\bids}{\mathbf{b}}
\newcommand{\bidc}{\mathbf{c}}
\newcommand{\bidd}{\mathbf{d}}
\newcommand{\bide}{\mathbf{e}}
\newcommand{\bidf}{\mathbf{f}}
\newcommand{\bidg}{\mathbf{g}}
\newcommand{\bidh}{\mathbf{h}}

\def \R {\mathbb{R}}
\newcommand{\E}{\mathbb{E}}
\newcommand{\matD}{\mathcal{D}}
\newcommand{\sdpgeq}{\succeq^+}
\newcommand{\sdpgt}{\succ^+}
\newcommand{\sdgeq}{\succeq}
\newcommand{\sdgt}{\succ}
\newcommand{\PO}{\text{Pareto efficiency}}
\newcommand{\fPO}{\mathrm{fPO}}
\newcommand{\rrpass}{\mathrm{RR}^\text{pass}}
\newcommand{\sdefoneplus}{\mathrm{SD}\text{-}\mathrm{EF}1^+}
\newcommand{\EFone}{\mathrm{EF}1}
\newcommand{\mech}{\mathcal{A}}
\newcommand{\desired}{\mathcal{M}^+}
\newcommand{\prefer}{\mathbf{p}}
\newcommand{\type}{\mathbf{t}}
\newcommand{\ia}{\textsc{IncrementalAccommodation}}
\newcommand{\rrp}{\textsc{IncrementalAccommodation}}
\newcommand{\se}{\textsc{SplitEqual}}

\newcommand{\qint}{q}
\newcommand{\qpos}{q^{\texttt{pos}}}

\newcommand{\fr}{f}
\newcommand{\ft}{f^*}

\maketitle

\begin{abstract}
    We study fair resource allocation with strategic agents. It is well-known that, across multiple fundamental problems in this domain, truthfulness and fairness are incompatible. For example, when allocating indivisible goods, no truthful and deterministic mechanism can guarantee envy-freeness up to one item (EF1), even for two agents with additive valuations. Or, in cake-cutting, no truthful and deterministic mechanism always outputs a proportional allocation, even for two agents with piecewise constant valuations. Our work stems from the observation that, in the context of fair division, truthfulness is used as a synonym for Dominant Strategy Incentive Compatibility (DSIC), requiring that an agent prefers reporting the truth, no matter what other agents report.

In this paper, we instead focus on Bayesian Incentive Compatible (BIC) mechanisms, requiring that agents are better off reporting the truth in expectation over other agents' reports. We prove that, when agents know a bit less about each other, a lot more is possible: BIC mechanisms can guarantee fairness notions that are unattainable by DSIC mechanisms in both the fundamental problems of allocation of indivisible goods and cake-cutting. We prove that this is the case even for an arbitrary number of agents,  as long as the agents' priors about each others' types satisfy a neutrality condition. Notably, for the case of indivisible goods, we significantly strengthen the state-of-the-art negative result for efficient DSIC mechanisms, while also highlighting the limitations of BIC mechanisms, by showing that a very general class of welfare objectives is incompatible with Bayesian Incentive Compatibility. Combined these results give a near-complete picture of the power and limitations of BIC and DSIC mechanisms for the problem of allocating indivisible goods.

\end{abstract}

\section{Introduction}

A central goal within the fair division literature is to design procedures that distribute indivisible or divisible goods among groups of agents. For a resource allocation procedure to reach fair and efficient outcomes, however, it needs to have access to the preferences of the participating agents. Whenever this information is private and each agent can strategically misreport it, the literature is riddled with negative results: eliciting the true preferences of the agents while simultaneously guaranteeing fairness and efficiency is usually impossible. One of the main underlying reasons is that monetary payments are commonly infeasible or undesired  in fair division, and designing ``truthful'' mechanisms in the absence of such payments poses often insurmountable obstacles.

Another factor that contributes to the plethora of impossibility results is that ``truthfulness'' in this literature has been used as a synonym for ``Dominant Strategy Incentive Compatibility'' (DSIC). This is a very demanding notion of incentive compatibility which requires that every agent prefers truth-telling to misreporting, \emph{no matter what the other agents' reports are}. 
This ensures that no agent will have an incentive to lie, even if they know \emph{exactly} what every other agent's preferences are. However, in most real-world applications it is unreasonable to assume that the agents know that much about each other, so this notion may be unnecessarily stringent. In fact, if we assume that an agent would strategically misreport their preferences only if they have enough information to suggest that this would be beneficial then, in some sense, the less an agent knows the less likely it is that they would misreport. Rather than going to the other extreme and assume the agents have no information about each other, in this paper we consider the well-established truthfulness notion of ``Bayesian Incentive Compatibility'' (BIC). 
This instead assumes that each agent's preference is drawn from a publicly known prior distribution and the requirement is that telling the truth is the optimal strategy for each agent \emph{in expectation over the other agents' reports}. Simply put, our goal in this paper is to understand what can and what cannot be achieved by BIC mechanisms in the context of fair division, and the extent to which they can outperform DSIC mechanisms.

\subsection{Our Contributions}

We explore the power and limitations of incentive compatible mechanisms when allocating either indivisible or divisible goods among $n$ strategic agents with additive valuations, i.e., such that an agent's value for a set of goods is equal to the sum of her values for each good separately.

\paragraph{Allocating indivisible goods.} It is well-known that the only deterministic DSIC mechanism for allocating indivisible goods among two agents in a Pareto efficient way is the patently unfair ``serial dictatorship'' mechanism~\cite{klaus2002strategy,schummer1996strategy}, where one of the agents is allocated \emph{all} the goods that they have a positive value for, and the other agent receives the leftovers.
Since Pareto efficiency may be a lot to ask for,
one could imagine that a more forgiving notion of efficiency based on stochastic dominance (SD)~\cite{BogomolnaiaMoulin01} may allow us to overcome this obstacle. Our main negative result shows that this is not the case: even for the more permissive notion of $\mathrm{SD}^+$efficiency (defined in Section~\ref{subsec:sdpoplus negative result}) instead of Pareto efficiency, we prove that the only deterministic DSIC mechanism remains the serial dictatorship, and this holds even if for every agent $i$, their value for each good $j$ can take one of only \emph{three} possible values, i.e. $v_{i,j} \in \{ 0 , x , y \}$ (\textbf{Theorem~\ref{theorem:sd_po_plus_result}}). In fact, we show that our stronger negative result is tight in more than one ways: if we were to further relax $\mathrm{SD}^+$efficiency to $\mathrm{SD}$ efficiency (also defined in Section~\ref{subsec:sdpoplus negative result}), or to further restrict the number of possible agent values to two instead of three, this would permit more deterministic DSIC mechanisms beyond serial dictatorship.

Surprisingly, it is also known that, even if we were to completely drop any efficiency requirement in order to avoid the extreme unfairness of the serial dictatorship, our ability to provide non-trivial fairness guarantees would still be severely limited. For example, if we wanted the outcome to be \emph{envy-free up to one good} (EF1), which is a well-studied relaxation of envy-freeness~\cite{LMMS2004,Budish2010}, this is known to be impossible to achieve using \emph{any} deterministic DSIC mechanism, even for instances with just two agents~\cite{amanatidis2017truthful}. Our main positive result shows that we can simultaneously overcome all of the aforementioned negative results if instead of DSIC we focus on BIC mechanisms. In fact, we can achieve this using a variation of the very practical Round-Robin procedure (where agents take turns choosing a single good each time). We prove that our variation, $\rrpass$, is BIC for any distribution over agent preferences that is ``neutral''~\cite{moulin1980implementing}, and it returns allocations that always combine fairness (in the form of EF1) with efficiency (in the form of $\mathrm{SD}^+$efficiency) (\textbf{Theorem~\ref{thm: rr main}}). 
Apart from providing a separation between BIC and DSIC mechanisms by bypassing both of these negative results, this also exhibits a third way in which our result from Theorem~\ref{theorem:sd_po_plus_result} is tight. Our definition of a ``neutral'' prior is satisfied by the common assumption in the stochastic fair division literature, that $v_{i,j}$s are drawn i.i.d. from an agent-specific distribution $\mathcal{D}_i$~\cite{bai2021envy,manurangsi2020envy,manurangsi2021closing,amanatidis2017approximation}, going beyond, it also allows for certain priors where valuations are neither independently nor identically distributed (see~\Cref{subsec:rr is bic}).

Encouraged by our positive result, which exhibits a practical BIC mechanism that combines fairness and $\mathrm{SD}^+$efficiency, one may wonder whether BIC mechanisms can even optimize well-motivated welfare functions. For example, a lot of recent work in fair division has focused on maximizing the Nash social welfare (the geometric mean of agents' utilities). Apart from guaranteeing Pareto efficiency, this outcome is also known to be EF1 for agents with additive valuations~\cite{caragiannis2019unreasonable}, and it is used in Spliddit, a popular platform for allocating indivisible goods~\cite{goldman2015spliddit,shah2017spliddit}. Our findings along this direction are negative: we show that even for a very simple neutral distribution (uniform over normalized valuations), a mechanism that maximizes any welfare function from a large well-studied family that satisfies the Pigou-Dalton principle (which, e.g., includes the Nash social welfare, the egalitarian social welfare, the utilitarian social welfare, and the leximin criterion) is not BIC (\textbf{Theorem~\ref{theorem:BIC-welfare-max}}). 
Finally, we prove that if we aim for the strong efficiency guarantee of fractional Pareto optimality (fPO), requiring that the outcome is not Pareto dominated even by randomized allocations, then BIC mechanisms cannot even satisfy \emph{fullfillment}, which is a fairness notion that is much weaker than EF1 (\textbf{Theorem~\ref{thm: fulfilling}}). Informally, an allocation is fulfilling if, whenever agent $i$ values at least $n$ goods strictly positively, her utility for her bundle is strictly positive. 

\paragraph{Allocating divisible goods.}
Moving beyond indivisible goods, we then consider mechanisms for the fair allocation of continuous, divisible, heterogeneous resources, using the classic cake-cutting model. The cake, which captures a divisible resource, is represented as the interval $[0,1]$ and it must be distributed among $n$ agents with different valuations on different parts of the interval. Typically, agents' valuation functions are described by probability density functions. 
Recently,~\cite{tao2022existence} showed that, even for $n=2$ agents with piecewise constant valuation functions, there is no deterministic and DSIC cake-cutting mechanism that always outputs a proportional allocation, i.e., one where each agent receives at least a $1/n$ fraction of their total value. Our main result  for this problem circumvents this impossibility for any number of agents by relaxing to BIC mechanisms (\textbf{Theorem~\ref{thm: positive for cake cutting}}). Specifically, we propose a deterministic cake-cutting mechanism that is proportional and BIC for all neutral priors. Our mechanism works over a sequence of $n$ rounds: in round $i$, agent $i$ arrives and agents $1$ through $i-1$ are asked to cut the pieces allocated to them so far into $i$ equal-sized and equal-valued crumbs.
Then, agent $i$ takes one crumb from each of the first $i-1$ agents. The fact that this operation can always be performed crucially relies on the result of~\cite{alon1987splitting} concerning the existence of a perfect partition.\footnote{Informally, a partition is said to be perfect if the value of every piece is the same for every agent.} Apart from piecewise constant valuation functions, our mechanism also provides the same guarantees for piecewise linear valuations, which can be succinctly represented, and for which perfect partitions can be computed efficiently. 

\subsection{Related Work}


Bayesian incentive compatibility in fair division has been considered in the random assignment problem, with $n$ goods and $n$ agents with \emph{ordinal} preferences.~\cite{dasgupta2022ordinal} consider ordinal BIC mechanisms (OBIC), and show that the probabilistic serial mechanism, which is ordinally efficient and satisfies equal treatment of equals, is OBIC with respect to the uniform prior. The notion of OBIC of~\cite{dasgupta2022ordinal} has been used in voting theory, dating back to~\cite{d1988ordinal}, to escape infamous dictatorship results~\cite{gibbard1973manipulation,satterthwaite1975strategy}.
For cardinal valuations, such as ours,~\cite{fujinaka2008bayesian} studies the case of a single indivisible good that must be allocated to one of $n$ agents whose (private) values for the good are drawn from a known product distribution. When payments are possible,~\cite{fujinaka2008bayesian} designs an envy-free, budget-balanced and BIC mechanism.

Another line of work considers other relaxations of DSIC. For example, in the random assignment problem,~\cite{mennle2021partial} introduce the notion of partial strategyproofness, a relaxation of DSIC, which is satisfied if truthful reporting is a dominant strategy for agents who have sufficiently different valuations for different objects, and show that in the context of school choice, this notion gives a separation between the classic and the adaptive Boston mechanism.
Starting with~\cite{troyan2020obvious}, a number of recent papers~\cite{ortega2022obvious,aziz2021obvious,psomasfair} relax the DSIC requirement and explore the design of not obviously manipulable (NOM) mechanisms. Under NOM, an agent reports their true type, unless lying is obviously better, where the definition of ``obvious'' for a strategy follows recent work in Economics~\cite{li2017obviously}. Closer to our interest here,~\cite{psomasfair} study the allocation of indivisible goods among additive agents, and prove that one can simultaneously guarantee EF1, PO, and NOM.~\cite{ortega2022obvious} show that, in the context of cake-cutting, NOM is compatible with proportionality: an adaptation of the moving-knife procedure satisfies both properties. 

Other ways to escape the aforementioned impossibility results in truthful fair division is restricting agents' valuations, e.g. by focusing on dichotomous~\cite{halpern2020fair,babaioff2021fair,benabbou2021finding,barman2021truthful} or Leontief valuations~\cite{ghodsi2011dominant,friedman2014strategyproof,parkes2015beyond}, or by using money-burning (that is, leave resources unallocated) as a substitute for payments, while trying to minimize the inefficiency that these payment substitutes introduce (e.g., \cite{hartline2008optimal,cole2013mechanism,fotakis2016efficient,friedman2019fair,abebe2020truthful}).
Finally, a recent research thread suggests the study of mechanisms that produce fair allocations in their equilibria~\cite{amanatidis2023round,amanatidis2023allocating}. Interestingly, for agents with additive valuations over indivisible goods, it is known that allocations obtained in equilibria of the Round-Robin algorithm are EF1 with respect to the agents’ true valuation functions~\cite{amanatidis2023allocating}.

\section{Preliminaries}
\textbf{Allocating indivisible goods.}
We are given a set $\items$ of $m$ indivisible items indexed by $[m] \coloneqq \{1, 2, \ldots, m\}$ to be allocated among a set $\agents$ of $n$ agents indexed by $[n]$. An allocation $\allo \in \{0,1\}^{n \times m}$ assigns items to agents such that $x_{i,j} = 1$ if agent $i$ gets item $j$ and $0$ otherwise. We use $\allo_i = (x_{i,1}, \dots , x_{i,m}) \in \{0,1\}^m$ to denote agent $i$'s \emph{allocation} and $X_i = \{j:x_{i,j}=1\}$ to denote agent $i$'s \emph{bundle} in $\allo$. In an allocation, each item $j \in \items$ must be allocated to exactly one agent, i.e., $\sum_{i=1}^n x_{i,j}=1$. Each agent $i \in \agents$ has a private type which is a valuation vector $\val_i \in \mathbb{R}^m_{\geq 0}$ where $v_{i,j}$ is agent $i$'s value for item $j \in \items$. Collectively, the valuations of all the agents are represented by a \emph{valuation profile} $\val = (\val_1, \ldots, \val_n)$. The \emph{utility} of agent $i$ for a given allocation $\allo$, denoted by $u_{i}(\allo)$, is \emph{additive} and defined as $u_{i}(\allo) \coloneqq \sum_{j \in \items} v_{i,j} x_{i,j}$. We will often overload notations and use $u_i(X_k)$ to denote the utility of agent $i \in \agents$ for a specific bundle $X_k \in \items$.

An allocation $\allo$ \emph{Pareto dominates} another allocation $\alloy$ if for all agents $i \in \agents$, we have $u_i(\allo) \geq u_i(\alloy)$ and $u_j(\allo) > u_j(\alloy)$ for some agent $j \in \agents$. An allocation $\allo$ is \emph{Pareto efficient} (or Pareto optimal) if there is no other allocation $\allo'$ that Pareto dominates it.

An allocation $\allo$ is called \emph{envy-free} (EF) if for every pair of agents $i$ and $j \in \mathcal{N}$, agent $i$ values their allocation at least as much as the allocation of agent $j$, i.e., $u_i(\allo_i) \geq u_i(\allo_j)$. However, for the indivisible item settings, an envy-free allocation is not guaranteed to exist. Therefore, we consider a relaxed notion called \emph{envy-free up to one good} (EF1). Formally, an allocation $\allo$ is EF1 if for every pair of agents $i$ and $j \in N$ with $X_j \neq \emptyset$, agent $i$ values their allocation at least as much as the allocation of agent $j$ without agent $i$'s favorite item, i.e., $u_i(X_i) \geq u_i(X_j\setminus \{g\})$ for some $g \in X_j$.

\textbf{Allocating divisible goods.}
We study the cake-cutting problem where the cake is represented by the interval $[0,1]$ and is to be allocated among the set $\agents$ of $n$ agents. An allocation $\alloi = (X_1,\dots, X_n)$ is a collection of mutually disjoint subsets of $[0,1]$ where $X_i \subseteq [0,1]$ denotes the subset of the cake allocated to agent $i \in \agents$. Each agent $i$'s private type is a density (valuation) function $f_i: [0,1] \rightarrow \R_{\geq 0}$ over the cake, such that $\int_0^1 f_i(x)dx =1$ for all $i \in \agents$. We use $\mathbf{f} = (f_1,\ldots,f_n)$ to denote the collection of density functions of all agents. Given a subset $S \subseteq [0,1]$, agent $i$'s utility on $S$ is defined as $u_i(S) = \int_{S} f_i(x)dx$.

We primarily focus on two classes of valuation functions, namely \emph{piecewise-constant} and \emph{piecewise-linear} valuations. For both families, the interval $[0,1]$ can be partitioned into finite intervals. For the former class, the value of the function in each interval is a constant; for the latter class it is linear.


An allocation $\alloi = (X_1,\ldots, X_n)$ is \emph{proportional} if each agent receives her average share of the entire cake. Formally, for each agent $i \in \agents$ we have $u_i(X_i) \geq \frac{1}{n}\cdot u_i([0,1]) = \frac{1}{n}$.

\textbf{Mechanisms and incentive compatibility.}
For both indivisible and divisible goods, each agent $i$ has a private type $t_i$: the valuations $\val_i$ for the first and the functions $f_i$ for the second. The agents, being strategic, may choose to misreport it if doing so increases their utility. We use $\bids = (\bids_1, \dots, \bids_n)$ to represent the reported type $\bids_i$ of all agent $i \in \agents$. A mechanism is represented by an allocation function $\allo(\cdot)$, which takes as input reports $\bids = (\bids_1, \dots, \bids_n)$ and outputs a feasible allocation $\allo(\bids)$. For notational simplicity, we often refer to a mechanism directly using its allocation function $\allo(\cdot)$.

We say a mechanism is \emph{incentive compatible} if it offers robustness guarantees against strategic behaviors of agents. In this paper, we consider two notions of incentive compatibility. A mechanism is \emph{dominant strategy incentive compatible (DSIC)} if truthful reporting is a dominant strategy for every agent. Formally, for every agent $i \in \agents$, every possible report $\bids_i$, and all possible reports of all other agents $\bids_{-i}$, we have
\[u_i(\allo(\type_i, \bids_{-i})) \geq u_i(\allo(\bids_i, \bids_{-i})).  \tag{DSIC}\]
We also study the incentives in the Bayesian setting where agents have some prior knowledge of each other's type. We assume that the type $\type_i$ of each agent $i \in \agents$ is drawn from some known prior distribution $\mathcal{D}_i$. A mechanism is \emph{Bayesian incentive compatible (BIC)} for priors $\times_{i=1}^n \mathcal{D}_i$ if reporting truthfully is a Bayesian Nash equilibrium, i.e., no agent can increase her \emph{expected} utility by unilaterally misreporting, the expectation being taken over the types of other agents. We use $\mathcal{D}_{-i}$ to denote the prior $\times_{j\neq i}\mathcal{D}_j$. Formally, for each agent $i \in \agents$, and every possible report $\bids_i$, we have
\[\mathop{\E}\limits_{\type_{-i} \sim \mathcal{D}_{-i}}[u_i(\allo(\type_i, \type_{-i})] \geq \mathop{\E}\limits_{\type_{-i} \sim \mathcal{D}_{-i}}[u_i(\allo(\bids_i, \type_{-i}))]. \tag{BIC}\]

\section{Allocating Indivisible Goods}\label{section:allocation}

Previous work on the allocation of indivisible goods has uncovered that DSIC mechanisms cannot be both fair and efficient. A notable example shows that requiring Pareto efficiency from a deterministic DSIC mechanism limits the available options to the patently unfair class of serial dictatorships. 

\begin{theorem}[{\cite{klaus2002strategy}}]\label{theorem:dictatorship} For $n=2$ agents, a DSIC deterministic mechanism is Pareto efficient if and only if it is a serial dictatorship.
\end{theorem}




Our first result, Theorem~\ref{theorem:sd_po_plus_result} in Section~\ref{subsec:sdpoplus negative result}, strengthens Theorem~\ref{theorem:dictatorship}, showing that even if we weaken the efficiency requirement from Pareto efficiency to $\mathrm{SD}^+$efficiency, a natural relaxation defined below, the only deterministic DSIC mechanisms that can guarantee $\mathrm{SD}^+$efficiency are serial dictatorships.
%
Another notable result shows that no deterministic DSIC mechanism always outputs EF1 allocations.
\begin{theorem}
[\cite{amanatidis2017truthful}]\label{EF1negative}
There is no deterministic DSIC mechanism that outputs EF1 allocations, even for instances involving just $n=2$ agents.
\end{theorem}
Theorem~\ref{thm: rr main} in Section~\ref{subsec:rr is bic} shows that if we relax the notion of truthfulness from DSIC to BIC, then neither one of these two negative results holds anymore, i.e., there exists a non-dictatorial deterministic BIC mechanism that guarantees EF1. 
In Section~\ref{subsec:pigou dalton negative} we also uncover the limitations of BIC mechanisms by proving that a wide family social choice functions are incompatible with BIC. 

\subsection{Strengthening the characterization of serial dictatorships}\label{subsec:sdpoplus negative result}

For a given agent $i$ let $\prefer_i$ be a ranking of items $j \in \items$ in decreasing order of $i$'s value for them, where $\prefer_i(k) = j$ if item $j$ is agent $i$'s $k^{\text{th}}$ favorite item.\footnote{For items with the same value, we break ties lexicographically, i.e., if $v_{i,j} = v_{i,k}$ and $j <k$, then item $j$ is ranked before $k$ in $\prefer_i$.} For two bundles $\allo_i, \alloy_i \in \{0,1\}^m$, $\allo_i$ \emph{weakly stochastically dominates $\alloy_i$ for agent $i$}, or $\allo_i \sdgeq_i \alloy_i$, iff 
for all $\ell \in [m]$ we have 
\begin{equation} \label{equation:sd-prefix-sums}
    \sum^{\ell}_{k = 1}  x_{i,\prefer_i(k)} \geq \sum^{\ell}_{k=1} y_{i,\prefer_i(k)}.
\end{equation}
Additionally, we say that $\allo_i$ \emph{(strictly) stochastically dominates $\alloy_i$ for agent $i$}, or $\allo_i \sdgt_i \alloy_i$, iff $\allo_i \sdgeq_i \alloy_i$ and $\allo_i \neq \alloy_i$. If $\allo_i \sdgt_i \alloy_i$, then there exists an $\ell \in [m]$ for which Inequality \eqref{equation:sd-prefix-sums} is strict.

To define our refinement of stochastic dominance, for each agent $i \in \agents$ we let $\desired_i$ denote the set of items that $i$ has a positive value for, i.e., $\desired_i \coloneqq \{j \in \items \ : \ v_{i,j} > 0 \}$. Let $m_i = |\desired_i|$. For two bundles $\allo_i, \alloy_i \in \{0,1\}^m$, we say that $\allo_i$ weakly stochastically \emph{and positively} dominates $\alloy_i$, denoted as $\allo_i \sdpgeq_i \alloy_i$, iff Inequality \eqref{equation:sd-prefix-sums} is satisfied for all $\ell \in \desired_i$. Thus, partial order $\sdpgeq_i$ is essentially the partial order $\sdgeq$ defined over the set of items $\desired_i$. 
 Similarly, $\allo_i \sdpgt_i \alloy_i$ iff $\allo_i \sdpgeq_i \alloy_i$ and $\allo_i \neq \alloy_i$. Akin to the previous definition, if $\allo_i \sdpgt_i \alloy_i$, then there exists an $\ell \in \desired_i$ for which Inequality \eqref{equation:sd-prefix-sums} is strict. Note that, if two bundles satisfy $\allo_i \sdgeq_i \alloy_i$, then by definition, $\allo_i \sdpgeq_i \alloy_i$ as well. The following lemma relates the partial orders defined above with utilities; it's proof along with other missing proofs from this subsection are deferred to Appendix \ref{appendix:sdpoplus missing proofs}.

 \begin{restatable}{lemma}{sdUtility}\label{proposition:sd-utility}
For any agent $i \in \agents$ and any two bundles $\alloa, \allob \in \{0,1\}^m$, the following statements hold: $(1)$ If $\alloa \sdpgeq_i \allob$, then $u_i(\alloa) \geq u_i(\allob)$, and $(2)$ If $\alloa \sdpgt_i \allob$, then $u_i(\alloa) > u_i(\allob)$.

\end{restatable}
 
We say that an allocation $\alloy$ \emph{stochastically dominates} 
$\allo$, denoted as $\alloy \sdgt \allo$ iff for all agents $i \in \agents$ we have $\alloy_i \sdgeq_i \allo_i$ and for at least one agent $j \in \agents$ we have $\alloy_j \sdgt_j \allo_j$. Similarly, we say that an allocation $\alloy$ \emph{stochastically and positively dominates} $\allo$, denoted as $\alloy \sdpgt \allo$ iff for all agents $i \in \agents$ we have $\alloy_i \sdpgeq_i \allo_i$ and for at least one agent $j \in \agents$ we have $\alloy_j \sdpgt_j \allo_j$. 
\begin{definition}
An allocation $\allo$ is $\mathrm{SD}$ efficient iff there is no other  allocation $\alloy$ such that $\alloy \sdgt \allo$.
\end{definition}
\begin{definition}
An allocation $\allo$ is $\mathrm{SD}^+$efficient iff there is no other  allocation $\alloy$ such that $\alloy \sdpgt \allo$.
\end{definition}
In the following lemma we show that $\sdpoplus$ is implied by $\PO$ and $\sdpoplus$ implies $\sdpo$. 

\begin{lemma}
\label{lemma:efficiency_comparison} The following two implications hold for any number of items and agents with additive preferences: (1) If an allocation $\allo$ is Pareto efficient then $\allo$ is $\mathrm{SD}^+$efficient as well; (2) If an allocation $\allo$ is $\mathrm{SD}^+$efficient then $\allo$ is $\mathrm{SD}$ efficient as well. Moreover, there exist allocations that are $\mathrm{SD}^+$efficient and not Pareto efficient, and allocations that are $\mathrm{SD}$ efficient and not $\mathrm{SD}^+$efficient.
\end{lemma}
\begin{proof}
    \emph{ Proving (1).} Towards a contradiction assume that allocation $\allo$ is Pareto efficient and not $\mathrm{SD}^+$efficient, i.e., there exists an allocation $\alloy$ such that $\alloy_i \sdpgeq_i \allo_i$ for all agent $i \in \agents$ and for some agent $j \in \agents$ we have $\alloy_j \sdpgt_j \allo_j$. Therefore, on applying Lemma \ref{proposition:sd-utility} we get that for all agents $i \in \agents$ we have $u_i(\alloy_i) \geq u_i(\allo_i)$ and $u_j(\alloy_j) > u_j(\allo_j)$. This contradicts the fact that $\allo$ is Pareto efficient and thus established the concerned implication.

    \emph{ Proving (2).} Assume, towards a contradiction, that there exists an allocation $\allo$ that is $\mathrm{SD}^+$efficient but not $\mathrm{SD}$ efficient. This implies the existence of another allocation $\alloy$ such that for each agent $i \in \agents$ we have $\alloy_i \sdgeq \allo_i$ and there exists an agent ${i^*} \in \agents$ such that $\alloy_{i^*} \sdgt_{i^*} \allo_{i^*}$. Note that, for each $i \in \agents$ we have $\alloy_i \sdgeq_i \allo_i$, which means that $\alloy_i \sdpgeq_i \allo_i$ as well. In particular, $\alloy_{i^*} \sdpgeq_{i^*} \allo_{i^*}$, however, it cannot be the case that $\alloy_{i^*} \sdpgt_{i^*} \allo_{i^*}$, because this would violate the fact that $\allo$ is $\mathrm{SD}^+$efficient. Therefore, the bundles received by agent $i^*$ must be such that $\alloy_{i^*} \sdgt_{i^*} \allo_{i^*}$ holds but $\alloy_{i^*} \sdpgt_{i^*} \allo_{i^*}$ doesn't hold.

The only way in which the bundles $\alloy_{i^*}, \allo_{i^*}$ can satisfy $\alloy_{i^*} \sdgt_{i^*} \allo_{i^*}$ but not $\alloy_{i^*} \sdpgt_{i^*} \allo_{i^*}$ is if $(i)$ for all $\ell \in [m_{i^*}]$, we have $\sum^{\ell}_{k = 1}  y_{{i^*},\prefer_{i^*}(k)} = \sum^{\ell}_{k=1} x_{{i^*},\prefer_{i^*}(k)}$ and $(ii)$ $\sum^{m}_{k = 1} y_{{i^*},\prefer_{i^*}(k)} > \sum^{m}_{k=1} x_{{i^*},\prefer_{i^*}(k)}$.
    Since the allocations $\alloy, \allo$ are integral, $(ii)$ equivalently implies that the number of items in $\alloy_{i^*}$ is strictly greater than the number of items in $\allo_{i^*}$. Since $\alloy_{i^*}$ has more items than $\allo_{i^*}$, it must be the case that for some other agent $t \in \agents \setminus \{i^*\}$, the bundle $\allo_t$ has more items than $\alloy_t$. This can be equivalently stated as $\sum^{m_t}_{k = 1} y_{t,\prefer_t(k)} < \sum^{m_t}_{k=1} x_{t,\prefer_t(k)}$.
    However, this inequality violates the fact that $\alloy_t \sdgeq_t \allo_t$, and concludes the proof.

    Finally, we show these implications are tight, by constructing allocations that are $\mathrm{SD}^+$efficient but not Pareto efficient, and $\mathrm{SD}$ efficient but not $\mathrm{SD}^+$efficient. These examples are provided in Appendix \ref{appendix:efficinecy_examples}.
\end{proof}

Now we state and prove the main result of this subsection.

\begin{restatable}{theorem}{sdpoplusResult}\label{theorem:sd_po_plus_result}\textup{(Characterization of serial dictatorship)}
For $n=2$ agents, a deterministic DSIC mechanism is $\mathrm{SD}^+$efficient if and only if it is a serial dictatorship. This holds even when agents have ternary valuations, i.e., for all agents $i$ and items $j$, we have $v_{i,j} \in \{0,x,y\}$ for some $y>x>0$.
\end{restatable}
\begin{proof}
Consider a setting where a fixed set of $m \geq 1$ items are to be divided among $n=2$ agents. Through an inductive argument will establish that every deterministic DSIC mechanism that outputs $\mathrm{SD}^+$efficient allocations for this setting is a serial dictatorship. 
Specifically, our approach will be as follows. We fix an arbitrary deterministic mechanism that is DSIC and outputs $\mathrm{SD^+}$efficient allocations. Then, we consider for each $d \in [m]$, the family of input bids $\mathcal{F}_d \coloneqq \{\bids = (\bids_1, \bids_2) \ | \ b_{i,j} = 0 \text{ for all agents } i \in [2] \text{ and items } j > d \}$. Note that, the family of input bids $\{ \mathcal{F}_d \}_{d \in [m]}$ follow a nested structure: $\mathcal{F}_1 \subsetneq \mathcal{F}_2 \ldots \subsetneq \mathcal{F}_{m-1} \subsetneq \mathcal{F}_m$. Furthermore, $\mathcal{F}_m$ represents the set of all possible input bids for the mechanism. We will inductively show that corresponding to the chosen mechanism there must exist one agent, who will be the dictator. That is, the mechanism on any input bid $\bids \in \mathcal{F}_m$ will allocate all the items positively-valued by the dictator, to the dictator, and the rest of the items will be allocated to the other agents.\footnote{The items which are valued zero (i.e., not desired) by both the agents will be allocated by this mechanism in some arbitrary deterministic manner.} Now we will formally state our induction hypothesis.

\noindent
{\bf Induction Hypothesis.} For a parameter $d \in [m]$, the induction hypothesis parameterized by $d$, $\mathrm{IH}(d)$ states that: the mechanism is a serial dictatorship for all input bids $\bids \in \mathcal{F}_d$.

In the subsequent proof, we will show that $\mathrm{IH}(d)$ holds true for all $d \in [m]$. Indeed, establishing $\mathrm{IH}(m)$ shows that the mechanism is a serial dictatorship since $\mathcal{F}_m$ represents the set of all input bids to the mechanism. We begin by establishing the base case.

\noindent
{\bf Base Case }(proving $\mathrm{IH}(1)$). Consider an input bid $\bids = (\bids_1, \bids_2) \in \mathcal{F}_1$ such that $b_{11} = b_{21} = x$ for some positive real number $x > 0$; since $\bids \in \mathcal{F}_1$, we have $b_{ij} = 0$ for all $i \in [n]$ and $j > 1$. Without loss of generality, we assume that our \emph{deterministic} mechanism outputs allocation $\allo(\bids)$ such that $x_{11}(\bids) = 1$ and $x_{12}(\bids) = 0$, i.e., item $1$ gets allocated to agent $1$.\footnote{Indeed, this is without loss of generality since we can always label the agent who receives item $1$ as being agent $1$.} To complete the base case, we will show that if the input bid $\widehat{\bids} \in \mathcal{F}_1$ is such that agent $1$ positively values item $1$, then irrespective of agent $2$'s reported bid, the mechanism will allocate item $1$ to agent $1$. This would establish that agent $1$ is a dictator for bids $\widehat{\bids} \in \mathcal{F}_1$, implying $\mathrm{IH}(1)$.

To this end, consider the report $\bids' = (\bids'_1, \bids_2) \in \mathcal{F}_1$ with $b'_{11} = y$. The fact that the mechanism is DSIC implies that the mechanism on input bids $\bids'$ it must continue to allocate item $1$ to agent $1$: otherwise, if item $1$ is not allocated to agent $1$, then agent $1$ can misreport her preference by bidding $\bids_1$ instead of $\bids'_1$, resulting in the bid profile $\bids = (\bids_1, \bids_2)$, in which case the mechanism will allocate item $1$ to agent $1$, leading to an increase in agent $1$'s utility, thereby contradicting the fact that the mechanism is DSIC. Thus, for both the bids $\bids$ and $\bids'$, item $1$ is not allocated to agent $2$ despite the fact that agent $2$ has a positive value for it. Therefore, the fact that the mechanism is DSIC implies that starting from either bids  $\bids$ or $\bids'$, agent $2$ cannot misreport her preferences in a way that item $1$ gets allocated to her, because this would lead to an increase in her utility and thus would contradict DSIC. This show that, if the input bids to the mechanism $\widehat{\bids} = (\widehat{b}_1, \widehat{b}_2) \in \mathcal{F}_1$ is such that agent $1$ positively-values item $1$ (i.e., $\widehat{b}_{1,1} \in \{x,y\}$), then irrespective of agent $2$'s bid, item $1$ is allocated to agent $1$. 

\noindent
{\bf Induction Step.} To complete the proof, we will now show that if $\IH(d)$ is true for some $1 \leq d < m$, then $\mathrm{IH}(d+1)$ will be true as well. Given $\IH(d)$, without loss of generality, we will assume that agent $1$ is the dictator for all instances $\widehat{\bids} \in \mathcal{F}_d$. To prove $\IH(d+1)$, we will consider a sequence of carefully designed bids. Then, utilizing the fact that the mechanism is DSIC and outputs $\mathrm{SD}^+$efficient allocations along with the induction hypothesis, $\IH(d)$, we will argue that the mechanism must behave like a serial dictatorship for these sequence of bids, eventually establishing $\IH(d+1)$. For ease of exposition, we will decompose the subsequent proof into the following claims.

\begin{claim}\label{claim:sdpop1}
    Let $\bidc  = (\bidc_1, \bidc_2) \in \mathcal{F}_{d+1}$ be the bid shown below,\footnote{The first and the second row of the matrix represent $\bidc_1$ and $\bidc_2$ respectively. For simplicity, the matrix representation doesn't show $c_{ij}$ for $i \in [2]$ and $j > d+1$ (columns $d+2$ to $m$ are omitted), because $\bidc \in \mathcal{F}_{d+1}$ and therefore $c_{i,j} = 0$.} where for all $j \leq d$, $c_{2,j} \in \{0,x,y\}$ are arbitrary fixed values. Given $\bidc$ as input, the mechanism must allocate all the first $d+1$ items to agent $1$, i.e., $\allo_{1,j}(\bidc) = 1$ for all $j \leq d+1$.
    \begin{align*}
    \bidc= \begin{pmatrix} \bidc_1 \\ \bidc_2 \end{pmatrix} =  \begin{pmatrix}
y & y & \cdots & y & x \\
c_{2,1} & c_{2,2} & \cdots & c_{2,d} & 0
\end{pmatrix}
\end{align*}
\end{claim}
\begin{proof}
 First, note that item $d+1$, which is positively valued only by agent $1$, must be allocated to agent $1$, i.e., $\allo_{1,d+1}(\bidc)=1$: if item $d+1$ is allocated to agent $2$, then it can be transferred to agent $1$, leading to a strict improvement in agent $1$'s bundle (wrt $\sdpgt_1$) and without affecting agent $2$'s bundle (wrt $\sdpgeq_2$) since item $d+1 \notin \desired_2$. That is, the allocation resulting from this transfer will stochastically and positively dominate $\allo(\bidc)$, contradicting the fact that the mechanism outputs $\mathrm{SD}^+$efficient allocations. Hence, item $d+1$ must be allocated to agent $1$. 

Now, towards a contradiction, suppose that there is an item $j \in [d]$ such that $\allo_{1,j}(\bidc) = 0$. This implies that agent $1$ is allocated the item $d+1$ and at most $d-1$ items from the bundle $[d]$. That is, her utility, $u_1(\bidc) \leq x + (d-1)y$. Now we will show that, in this case, agent $1$ can improve her utility by misreporting her preference to be $\bidc'_1$ where $c'_{1,j} = y$ for all items $j \in [d]$ and $c'_{1,j} = 0$ for items $j > d$. Note that, if agent $1$ misreports to $\bidc'_1$, then the bid profile $\bidc' = (\bidc'_1, \bidc_2) \in \mathcal{F}_d$ since for all items $j > d$, we have $c'_{1,j} = c_{2,j} = 0$. Because $\bidc' \in \mathcal{F}_d$, we can invoke $\mathrm{IH}(d)$ to infer that, agent $1$ will be the dictator, and thus all the items $j \in [d]$ will be allocated to agent $1$. Hence, her utility, as per her true preference, $\bidc_1$, in $\allo(\bidc')$ will be $u_i(\bidc') \geq d \cdot y$. However, this violates the fact that the mechanism is DSIC, since agent $1$ is able to improve her utility from $u_1(\bidc) \leq x + (d-1)y$ to $u_i(\bidc') \geq d \cdot y$ by misreporting because $y > x$ and hence $u_i(\bidc') \geq d \cdot y > x + (d-1)y \geq u_1(\bidc)$. Therefore, all the items $j \in [d+1]$ must be allocated to agent $1$, $\allo_{1,j}(\bidc) = 1$. 
\end{proof}

\begin{claim}\label{claim:sdpop2}
    Let $\bide  = (\bide_1, \bide_2) \in \mathcal{F}_{d+1}$ be the bid shown below. Given $\bide$ as input, the mechanism must allocate all the first $d+1$ items to agent $1$, i.e., $\allo_{1,j}(\bide) = 1$ for all $j \leq d+1$.
\begin{align*}
    \bide = \begin{pmatrix} \bide_1 \\ \bide_2 \end{pmatrix} =  \begin{pmatrix}
x & x & \cdots & x & y \\
y & y & \cdots & y & x
\end{pmatrix}
\end{align*}
\end{claim}


\begin{proof}
Consider the profile $\bidf = (\bidf_1, \bidf_2) \in \mathcal{F}_{d+1}$ shown below where for all items $j \in [d+1]$ we have $f_{1,j} \in \{x, y\}$, $f_{2,j} \in \{0,x,y\}$ for all $j \in [d]$ and $f_{2,d+1} = 0$.

\begin{align*}
    \bidf = \begin{pmatrix} \bidf_1 \\ \bidf_2 \end{pmatrix} =  \begin{pmatrix}
    f_{1,1} & f_{1,2} & \cdots & f_{1,d} & f_{1, d+1} \\
    f_{2,1} & f_{2,2} & \cdots & f_{2,d} & 0
\end{pmatrix}
\end{align*}

By definition, only items in $[d+1]$ are positively-valued by agent $1$, i.e., $\items^+_1 = [d+1]$. We will first show that all items in $[d+1]$ will be allocated to agent $1$ in $\allo(\bidf)$. Towards this, note that if agent $1$ doesn't receive all the items in $[d+1]$, then she can increase her utility by misreporting her preference from $\bidf_1$ to $\bidc_1$: if agent $1$ misreports, then the resultant bids $(\bidc_1, \bidf_2)$ will match the condition in Claim \ref{claim:sdpop1}, and hence, by Claim \ref{claim:sdpop1}, all the items in $[d+1]$ will be allocated to agent $1$ resulting in an increase in utility. This will violate the fact that the mechanism is DSIC, implying that all the items in $[d+1]$ must be allocated to agent $1$ in the allocation $\allo(\bidf)$. In particular, if we consider the following bid $\bidf' = (\bidf'_1, \bidf'_2) \in \mathcal{F}_{d+1}$, which is a particular instantiation of the bids in $\bidf$, then we know that all the items in $[d+1]$ will be allocated to agent $1$.
\begin{align*}
    \bidf' = \begin{pmatrix} \bidf'_1 \\ \bidf'_2 \end{pmatrix} =  \begin{pmatrix}
x & x & \cdots & x & y \\
y & y & \cdots & y & 0
\end{pmatrix}
\end{align*}
Using this, we will now show that, the items $[d+1]$ must be allocated to agent $1$ in the allocation $\allo(\bide)$ corresponding to the concerned bids $\bide$ shown below.
\begin{align*}
    \bide = \begin{pmatrix} \bide_1 \\ \bide_2 \end{pmatrix} =  \begin{pmatrix}
x & x & \cdots & x & y \\
y & y & \cdots & y & x
\end{pmatrix}
\end{align*}



 The subsequent proof will be in two parts: first, we will show that all the items in $[d]$ will be allocated to agent $1$ in $\allo(\bide)$, then, we will show the same of item $d+1$. For the first part, assume towards a contradiction, that agent $2$ receives an item $j^* \in [d]$ for bid $\bide$, i.e., $\allo_{2,j^*}(\bide) = 1$. Given this, starting from the input bids $\bidf' = (\bidf'_1, \bidf'_2)$, agent $2$, having a true valuation of $\bidf'_2$, can increase her utility by misreporting to $\bide_2$. Note that, if agent $2$ misreports to $\bide_2$, then the resultant bids $(\bidf'_1, \bide_2)$ is exactly same as the bids $\bide$. Therefore, if agent $2$ misreports, then she will be allocated the item $j^* \in [d]$, since $\allo_{2,j^*}(\bide) = 1$ where $j^*$ is positively valued by agent $2$, $\bidf'_{2,j^*} = y > 0$. On the contrary, if agent $2$ reports her true valuation, then we know that in the resultant bids, $\bidf' = (\bidf'_1, \bidf'_2)$, all the items in $[d+1]$ will be allocated to agent $1$. Hence, by misreporting to $\bide_2$, agent $2$ can increase her utility, contradicting the fact that the mechanism is DSIC. Therefore, for all items $j \in [d]$, we must have $\allo_{1,j}(\bide) = 1$. Finally, we will show that $\allo_{1,d+1}(\bide) = 1$ as well. Suppose this is not the case, i.e., $\allo_{1,d+1}(\bide) = 0$. As already shown, item $d$ is allocated to agent $1$, and as per our assumption, item $d+1$ is allocated to agent $2$. Note that, in bid $\bide$, the value of item $d$ for agent $1$ and item $d+1$ for agent $2$ is $x$, whereas, the value of item $d+1$ for agent $1$ and item $d$ for agent $2$ is $y$ where $y>x$. Therefore, if agents $1$ and $2$ exchange items $d$ and $d+1$ between them, then the bundles of both the agents in the resultant allocation are strictly better wrt $\sdpgt_1$ and $\sdpgt_2$ respectively. This implies that the allocation before the exchange was not $\mathrm{SD}^+$efficient, contradicting the fact that the mechanism outputs $\mathrm{SD}^+$efficient allocations. Therefore, item $d+1$ also must be allocated to agent $1$, i.e., $\allo_{1,d+1}(\bide) = 1$.
\end{proof}

Finally, we will show that $\IH(d+1)$ holds. We begin by considering the bids $\bidg = (\bidg_1, \bidg_2)$ shown below, where $g_{2,j} \in \{0,x,y\}$ for all items $j \in [d+1]$, i.e., agent $2$'s bids are completely arbitrary.
\begin{align*}
    \bidg = \begin{pmatrix} \bidg_1 \\ \bidg_2 \end{pmatrix} =  \begin{pmatrix}
x & x & \cdots & x & y \\
g_{2,1} & g_{2,2} & \cdots & g_{2,d} & g_{2, d+1} \end{pmatrix}
\end{align*}

The allocation $\allo(\bidg)$ must be such that all items $[d+1]$ are allocated to agent $1$: otherwise, if there is an item $j' \in [d+1]$ such that $\allo_{2,j'}(\bidg) = 1$, then starting from the bids $\bide$, agent $2$ can misreport her valuation to $\bidg_2$ and increase her utility. As shown in Claim \ref{claim:sdpop2}, agent $2$ does not receive any item from $[d+1]$ in the allocation $\allo(\bide)$ despite having positive value for all of them in $\bide$. However, if starting from profile $\bide$, agent $2$ misreports her value to be $\bidg_2$, then as per our assumption, she will at least receive item $j' \in [d+1]$ in the resultant profile $(\bide_1, \bidg_2)$ because $\bidg = (\bide_1, \bidg_2)$, thereby increasing her utility. Since this violates the fact that the mechanism is DSIC, it must be the case that all items in $[d+1]$ are allocated to agent $1$. This implies that if agent $1$ reports valuation $\bidg_1$, then irrespective of agent $2$'s bid all the items in $[d+1]$ are allocated to agent $1$.

As the final step of the proof, we consider a completely arbitrary bid profile $\bidh = (\bidh_1, \bidh_2) \in \mathcal{F}_{d+1}$ shown below where $h_{i,j} \in \{0, x, y\}$ for all agents $i \in [2]$ and items $j \in [d+1]$. Since the mechanism is DSIC, we have that in allocation $\allo(\bidh)$, agent $1$ must receive all the items she has a positive value for: since otherwise she can misreport her value as $\bidg_1$, resulting in the bid profile $(\bidg_1, \bidh_2)$ (which is equivalent to $\bidg$), and as previously shown, agent $1$ receives all $[d+1]$ items in allocation $\bidg$. Therefore, in an arbitrary bid profile, $\bidh \in \mathcal{F}_{d+1}$, agent $1$ will receive all the items she has a positive value for, establishing her to be the dictator and proving $\IH(d+1)$. Finally, note that the items which are positively valued only by agent $2$, must be allocated to agent $2$, since the mechanism outputs $\mathrm{SD}^+$efficient allocations. This completes the induction step and concludes the proof.

\begin{align*}
    \bidh = \begin{pmatrix} \bidh_1 \\ \bidh_2 \end{pmatrix} =  \begin{pmatrix}
h_{1,1} & h_{1,2} & \cdots & h_{1,d} & h_{1, d+1} \\
h_{2,1} & h_{2,2} & \cdots & h_{2,d} & h_{2, d+1} \end{pmatrix} \quad 
\end{align*}
\end{proof}

\paragraph{Tightness of the above characterization}

In the following two propositions, we show that Theorem~\ref{theorem:sd_po_plus_result} is tight in two orthogonal ways (their proof appear in~\Cref{appendix:sdpoplus missing proofs}). Specifically, if we relax $\sdpoplus$ to the weaker notion of $\sdpo$, then serial dictatorship no longer remains the only DSIC mechanism --- this is formally stated in Proposition \ref{proposition:sdpo-dict}. Note that although $\mathrm{SD}$ efficiency is widely used in matching (where all agents receive the same number of items), the notion of $\mathrm{SD}^+$efficiency is more appropriate when each agent may receive a different number of items. Specifically, an allocation can only be stochastically dominated by another allocation if all agents receive at least the same number of items; this is not true for stochastically and positively dominated.


\begin{restatable}{proposition}{tightnessPropOne}\label{proposition:sdpo-dict}
    For $n=2$ agents, there exists a deterministic mechanism other than the serial dictatorship that is DSIC and outputs $\mathrm{SD}$ efficient allocations.
\end{restatable}

Second, the proof of Theorem \ref{theorem:sd_po_plus_result} holds even for ternary preferences i.e., $v_{i,j} \in \{0,x,y\}$ for all $i \in \agents$ and $j \in \items$. In fact, Theorem \ref{theorem:sd_po_plus_result} does not hold if we consider the slightly more restricted class of bivalued preferences, i.e., $v_{i,j} \in \{x,y\}$ for fixed $x,y \in \mathbb{R}_{\geq 0}$. 

\begin{restatable}{proposition}{bivaluedPref}\label{proposition:b prefs}
    For bivalued additive preferences, there exist DSIC, non-dictatorial mechanisms that output Pareto efficient allocations.
\end{restatable}



\subsection{Round Robin is BIC}\label{subsec:rr is bic}

We now show that a variation of Round-Robin, which we call \textsc{Round-Robin}$^{\text{pass}}$ ($\rrpass$), is BIC for neutral priors. Moreover, $\rrpass$ always outputs an $\mathrm{SD}^+$efficient and $\EFone$ allocation,\footnote{We note that items which have zero value for all the agents won't be allocated by $\rrpass$. Indeed, this does not affect the fairness and efficiency guarantees of $\rrpass$.} bypassing the negative results of Theorems~\ref{EF1negative} and~\ref{theorem:sd_po_plus_result}.


\paragraph{The Bayesian Setting.} 
We assume that the for each agent $i \in \agents$, there is a \emph{neutral} prior distribution $\mathcal{D}_i$ supported over $\mathbb{R}_{\geq 0}^m$ such that the valuation vector $\val_i = (v_{i,1}, \ldots, v_{i,m}) \sim \mathcal{D}_i$. We say that $\mathcal{D}_i$ is neutral if for every permutation $\sigma: [m] \to [m]$, we have $\Pr_{\val_i \sim \mathcal{D}_i}[v_{i,\sigma(1)} > \ldots > v_{i,\sigma(m)}] = \frac{1}{m!}$, and for each $j \in \items$, the marginal distribution of $\mathcal{D}_i$ wrt $v_{ij}$ does not have point masses. Intuitively, neutral priors imply that all items are equivalent, in expectation. A natural example of a neutral prior is $v_{i,j} \sim \widehat{\mathcal{D}}_i$ where $\widehat{\mathcal{D}}_i$ is supported on $\mathbb{R}_{\geq 0}$ and does not have a point mass. This exact prior has been extensively studied in the literature on stochastic fair division \cite{bai2021envy,manurangsi2020envy,manurangsi2021closing,amanatidis2017approximation}; a stochastic fair division model based on neutrality is also studied in~\cite{benade2023existence}. 
Note that, neutral priors are permissive enough to allow for correlated item values, for e.g., if $\val_i \sim \Delta^{m-1}$ where $\sim \Delta^{m-1}$ denotes a uniform draw from the $m-1$-simplex,\footnote{Recall that, an $m-1$-simplex $\Delta^{m-1} = \{(x_1, x_2, \ldots, x_m) : \sum_{i} x_i = 1 \text{ and } x_i \geq 0 \text{ for all } i\}$} then the item values are identically distributed but not independent. Additionally, the item values for an agent don't even have to be identically drawn: for two items, the prior $v_{i,1}\sim \mathcal{U}[0,1]$ and $v_{i,2}\sim \mathcal{U}[\frac{1}{2},\frac{3}{4}]$ is neutral.\footnote{$\mathcal{U}[a,b]$ represents a uniform distribution over the interval $[a,b]$.}

\paragraph{Our Mechanism.}
We begin by defining the $\rrpass$ mechanism: intuitively, we execute the standard Round-Robin procedure but allow agents to \emph{pass} their turn if all the remaining items are zero-valued. Formally, given the report $\bids_i$ for each agent $i \in \agents$, the mechanism first computes a tuple $(\prefer_i, m_i)$, where $\prefer_i$ is the strict preference order of items for agent $i$ (ties are broken lexicographically), and $m_i = |\items^+_i|$, where $\items^+_i = \{j : b_{i,j}>0\}$ is the set of items agent $i$ values positively. 
Then the mechanism constructs the output allocation over rounds. In each round, agents arrive in the fixed order $1, 2, \ldots, n$, and each agent $i \in \agents$ is allocated the item with the smallest (according to $\prefer_i$) ranking among the set of remaining items; if in agent $i$'s turn all items in $\items_i^+$ have already been allocated, then $\rrpass$ skips $i$'s turn. This process is repeated until either all the items are allocated or all remaining items are undesirable for all the agents. In the latter case, the remaining items are allocated deterministically and arbitrarily among the agents, resulting in a complete allocation. $\rrpass$ is almost ordinal, i.e., to execute it we need the (strict) preference order $\prefer_i$, but also the number of zero-valued items $m_i$ for each $i \in \agents$ (noting that one can construct $\items^+_i$ by looking at the bottom $m_i$ elements of $\prefer_i$). Thus, in the subsequent discussion, we often interchange $\val_i$ with $(\prefer_i, m_i)$.
Our goal is to prove the following theorem.

\begin{theorem}\label{thm: rr main}
   The \textsc{Round-Robin}$^{\text{pass}}$ mechanism is Bayesian incentive compatible for neutral priors, and always outputs SD$^+$efficient and $\EFone$ allocations. 
\end{theorem}

First, we prove that $\rrpass$ is BIC for neutral priors. The high-level idea of our proof is similar to the arguments of~\cite{dasgupta2022ordinal} for showing that the probabilistic serial mechanism is BIC in a setting with complete ordinal preferences. Unlike the ordinal setting, our cardinal setting allows agents to have undesirable, zero-valued items. As hinted by Theorem~\ref{theorem:sd_po_plus_result} and Proposition~\ref{proposition:sdpo-dict}, allowing for zero values makes it harder to establish incentive guarantees. To mitigate this, our mechanism $\rrpass$ (contrary to Round-Robin) never allocates an item to an agent that doesn't value it, unless the item is valued at zero by everyone. This special treatment of zero values is required to prove key structural properties about the \emph{interim allocation} (defined below), which are used throughout our proofs.





For each agent $i \in \agents$ and item $j \in \items$ we define the \textbf{interim allocation} $\qint_{i,j}(\bids_i) \in [0,1]$ of a mechanism with allocation function $\allo(\cdot)$ to be the probability that agent $i$ is allocated item $j$ when she reports $\bids_i$. Formally, $\qint_{i,j}(\bids_i) = \mathbb{P}_{\val_{-i} \sim \mathcal{D}_{-i}} [x_{i,j}(\bids_i, \val_{-i})]$.

We show that, for each agent $i \in \agents$, the interim allocation of the $\rrpass$ mechanism, given neutral priors, is very structured. 
For each agent $i \in \agents$, there is a (fixed) \textbf{positional interim allocation} $\mathbf{q_i^{\texttt{pos}}} = (\qpos_{i,1}, \qpos_{i,2}, \ldots, \qpos_{i,m})$ such that, if agent $i$ reports valuation $\bids_i$ to $\rrpass$, then her interim allocation $\qint_{i,j}(\bids_i)$ is exactly $\qpos_{i,k}$ if $b_{i,j} > 0$ and $\prefer_i(k) = j$; otherwise, $\qint_{i,j}(\bids_i) = 0$ if $b_{i,j}= 0$ (Lemma~\ref{lemma:itemindependence}). Intuitively, the interim allocation of a positively valued item depends on only its position in the preference order. Furthermore, the positional interim allocations are monotone: $\qpos_{i,1} \geq \ldots \geq \qpos_{i,m}$ (Lemma~\ref{lemma:interimmonotonicity}). 

Since $\rrpass$ only requires the tuples $\{(\prefer_i, m_i)\}_{i \in \agents}$ to execute, overloading notations, we will also write the interim allocation as a function of these tuples, i.e., $q_{i,j}(\prefer_i,m_i)$. 
As previously mentioned, neutral priors ensure that each of the $m!$ possible ordering of $\prefer_j$, for any agent $j \neq i$ is equally likely. In addition, the fact that each marginal distribution of $\mathcal{D}_{j}$ does not have any point mass, implies that, with probability $1$, we have $m_j = m$ for all $j \neq i$. Using these observations, we can write the interim allocation of agent $i$ as,
\begin{align}\label{eq:interimasorder}
    q_{i,j}(\prefer_i,m_i) = \frac{1}{\left(m!\right)^{n-1}} \sum_{\prefer_{-i}}x_{i,j}((\prefer_i, m_i), (\prefer_{-i},m)),
\end{align}
for each item $j \in \items$, where $(\prefer_{-i},m) = \{(\prefer_j, m_j)\}_{j \neq i}$.

\begin{restatable}{lemma}{ItemIndependence}\label{lemma:itemindependence}
For each agent $i\in \agents$, there exists a positional interim allocation $\mathbf{q_i^{\texttt{pos}}} = (\qpos_{i,1}, \ldots, \qpos_{i,m})$ such that if agent $i$ reports any valuation $\bids_i$, then her interim allocation $\qint_{i,j}(\bids_i) = \qpos_{i,k}$ if (i) $b_{i,j} > 0$, and (ii) $\prefer_i(k) = j$; otherwise, $\qint_{i,j}(\bids_i) = 0$ if $b_{i,j}= 0$.
\end{restatable}
\begin{proof}
By definition of $\rrpass$, if $b_{i,j} = 0$, then agent $i$ can be allocated item $j$ only if $j$ is valued zero by all other agents. Since the marginal distributions of $\{\mathcal{D}_j\}_{j \neq i}$ do not have a point mass, this is a zero probability event, and hence, we get that $q_{i,j}(\bids_i) = 0$ if $b_{i,j} = 0$. In the subsequent proof, we will focus on the case when $b_{i,j} > 0$. Let $\widehat{\bids}_i$ and $\widetilde{\bids}_i$ be two bids whose corresponding preference orders are $\widehat{\prefer}_i$ and $\widetilde{\prefer}_i$ respectively. Additionally, suppose $\widehat{m} = |\{j : \widehat{b}_{i,j} > 0 \}|$ and $\widetilde{m} = |\{j : \widetilde{b}_{i,j} > 0 \}|$ for some $\widehat{m} \leq \widetilde{m}$. Given this, we will show that if $\widehat{\prefer}_i(k) = r$ and $\widetilde{\prefer}_i(k) = s$ for some $k \leq \widehat{m}$, then we must have $q_{i,r}(\widehat{\prefer}_i, \widehat{m}) = q_{i,s}(\widetilde{\prefer}_i, \widetilde{m})$. This will imply that the interim allocation only depends on the ranking of the item. Finally, defining $\qpos_{i,k} \coloneqq q_{i,r}(\widehat{\prefer}_i, \widehat{m}) = q_{i,s}(\widetilde{\prefer}_i, \widetilde{m})$ will complete the proof.

The proof of this will be based on establishing the following, simpler, pairwise-version of the above proposition. For a given preference order $\prefer_i$, define $\prefer'_i = \sigma(\prefer_i)$ to be the preference order obtained by swapping the ranks of items $a$ and $b$, where $\sigma:[m]\to [m]$ is some permutation of items with $\sigma(a)= b$, $\sigma(b) = a$, and $\sigma(\ell) = \ell$ for all $\ell \in [m] \setminus \{a,b\}$. 
Given this, we will show that the interim allocations corresponding to $\prefer_i$ and $\prefer'_i$ are such that, we have 

\begin{align}\label{eq:pairwise}
q_{i,j}(\prefer_i,\widehat{m}) = q_{i,\sigma(j)}(\prefer'_i,\widehat{m}),
\end{align}

for all items $j \in \items$. This implies that the interim allocations do not change wrt the ranking of items if we swap two items. Note that, by repeatedly swapping pairs of items, we can convert the ordering $\widehat{\prefer}_i$ to $\widetilde{\prefer}_i$. By Equation (\ref{eq:pairwise}), this implies that $q_{i,r}(\widehat{\prefer}_i, \widehat{m}) = q_{i,s}(\widetilde{\prefer}_i, \widehat{m})$. Additionally, note that $q_{i,s}(\widetilde{\prefer}_i, \widehat{m}) = q_{i,s}(\widetilde{\prefer}_i, \widetilde{m})$. This follows the fact that $k \leq \widehat{m} \leq \widetilde{m}$ (where $\widetilde{\prefer}_i(k) = s$) and by definition of $\rrpass$ increasing the second value of the tuple (i.e., $\widehat{m} \leq \widetilde{m}$) cannot change the allocation of an item whose rank is less than or equal to $\widehat{m}$. Hence, Equation (\ref{eq:pairwise}), once established will imply that $q_{i,r}(\widehat{\prefer}_i, \widehat{m}) = q_{i,s}(\widetilde{\prefer}_i, \widetilde{m})$.

%



Therefore, to complete the proof will now establish Equation (\ref{eq:pairwise}). Note that for every possible report of other agents, $\prefer_{-i}$ we have,
\begin{align*}\label{eq:abswap}
    x_{i,j}((\prefer_i,\widehat{m}) (\prefer_{-i},m)) & = x_{i,\sigma(j)}((\sigma(\prefer_i),\widehat{m}),(\sigma(\prefer_{-i}),m))\\
    & = x_{i,\sigma(j)}((\prefer'_i,\widehat{m}),(\sigma(\prefer_{-i}),m)), \numberthis
\end{align*}
 where $(\sigma(\prefer_{-i}),m) = \{(\sigma(\prefer_{j}),m)\}_{j \neq i}$. The first equality above is based on the following observation. Item $j$ is allocated to agent $i$ given preference orders $\prefer_i$ and $\prefer_{-i}$ iff agent $i$ receives the item $\sigma(j)$ given preference orders $\sigma(\prefer_i)$ and $\sigma(\prefer_{-i})$. This is true because, given the preference orders of each agent, the allocation returned by $\rrpass$ depends only on the preference orders and not the indexes of the items. Using this, we have that for all items $j \in \items$,
\begin{align*}
    q_{i,j}(\prefer_i,\widehat{m}) &= \frac{1}{(m!)^{n-1}}\sum_{\prefer_{-i}}x_{i,j}((\prefer_i,\widehat{m}), (\prefer_{-i},m)) \tag{via \eqref{eq:interimasorder}}\\
    & = \frac{1}{(m!)^{n-1}}\sum_{\prefer_{-i}} x_{i,\sigma(j)}((\prefer'_i,\widehat{m}),(\sigma(\prefer_{-i}),m)) \tag{via \eqref{eq:abswap}}\\
    &= \frac{1}{(m!)^{n-1}}\sum_{\prefer_{-i}} x_{i,\sigma(j)}((\prefer'_i,\widehat{m}),(\prefer_{-i},m)) \\
    & = q_{i,\sigma(j)}(\prefer'_i,\widehat{m}) \tag{via \eqref{eq:interimasorder}},
\end{align*}
where the penultimate equality holds because the set of all $\prefer_{-i}$ and the set of all $\sigma(\prefer_{-i})$ is the same. This establishes Equation (\ref{eq:pairwise}) and concludes the proof.
\end{proof}


\begin{restatable}{lemma}{interimMonotonicity}\label{lemma:interimmonotonicity}
Positional interim allocations are monotone, i.e., $\qpos_{i,1} \geq \qpos_{i,2} \geq \ldots \geq \qpos_{i,m}$ for all $i \in \agents$.
\end{restatable}
\begin{proof}
 We will show that $\qpos_{i,k} \geq \qpos_{i,k+1}$ for all $k < m$. Towards this, consider a bid $\bids_i$ such that $(i)$ all items are positively valued, $|\desired_i| = m$, and $(ii)$ the corresponding preference order $\prefer_i$ is such that $\prefer_i(k) = a$ and $\prefer_i(k+1) = b$ for some items $a,b \in \items$. Additionally, let $\prefer'_i$ be a preference order constructed by swapping $a$ and $b$ in $\prefer_i$, i.e., $\prefer'_i(k) = b$ and $\prefer'_i(k+1) = a$, and $\prefer'_i(\ell) = \prefer_i(\ell)$ for all $\ell \neq \{k,k+1\}$.
 
 As the key part of the proof, we will first show that inequality 
 \begin{align}\label{ineq:e_mono}
 x_{i,a}((\prefer_i,m),(\prefer_{-i},m)) \geq x_{i,a}((\prefer'_i,m), (\prefer_{-i},m)),
 \end{align}
 where $(\prefer_{-i},m) = \{(\prefer_j,m)\}_{j \neq i}$ holds for all preference orders of other agents $\prefer_{-i}$.\footnote{We note that this property has been referred to as elementary monotonicity in \cite{dasgupta2022ordinal}.} Note that, both the quantities are binary, $x_{i,a}((\prefer'_i,m), (\prefer_{-i},m)), x_{i,a}((\prefer_i,m),(\prefer_{-i},m)) \in \{0,1\}$. Hence, the concerned inequality is trivially satisfied if $x_{i,a}((\prefer_i,m),(\prefer_{-i},m)) = 1$. Otherwise if $x_{i,a}((\prefer_i,m),(\prefer_{-i},m)) = 0$, we can show that necessarily $x_{i,a}((\prefer'_i,m), (\prefer_{-i},m)) = 0$ as well. This follows from the following observations. Since agent $i$ does not receive item $a$ when reporting preference $\prefer_i$, $a$ must have already been allocated to some other agent before agent $i$ could have picked it. Now if she reports $\prefer'_i$, wherein item $a$ has an even lower rank ($\prefer'_i(k+1) = a$), then agent $i$ still won't receive item $a$ because similar to the previous case, $a$ would have already been allocated to the same agent given that $\prefer_i$ and $\prefer'_i$ only differ in ranks $k$ and $k+1$, and preferences of other agents $\prefer_{-i}$ remain unchanged. Thus, if $x_{i,a}((\prefer_i,m),(\prefer_{-i},m)) = 0$ then we also have $x_{i,a}((\prefer'_i,m), (\prefer_{-i},m)) = 0$, establishing Inequality (\ref{ineq:e_mono}). 

 Now using this, we get the desired inequality.

 \begin{align*}
     \qpos_{i,k} & = q_{i,a}(\prefer_i,m) \tag{Using Lemma \ref{lemma:itemindependence} and $\prefer_i(k) = a$}\\
     &= \frac{1}{\left(m!\right)^{n-1}}\sum_{\prefer_{-i}} x_{i,a}((\prefer_i,m),(\prefer_{-i},m)) \tag{via (\ref{eq:interimasorder})}\\
     &\geq \frac{1}{(m!)^{n-1}}\sum_{\prefer_{-i}} x_{i,a}((\prefer'_i,m), (\prefer_{-i},m)) \tag{via (\ref{ineq:e_mono})}\\
     & = q_{i,a}(\prefer'_i,m) \tag{via (\ref{eq:interimasorder})}\\
     & = \qpos_{i,k+1} \tag{Using Lemma \ref{lemma:itemindependence} and $\prefer'_i(k+1) = a$}
 \end{align*}

\end{proof}

Using the previous two lemmas we now show that $\rrpass$ is BIC for neutral priors. The high-level intuition is as follows. If agent $i$ reports valuation $\bids_i$ having a corresponding preference order $\prefer_i = (j_1, j_2, \ldots, j_m)$, then Lemma~\ref{lemma:itemindependence} implies that the interim allocations $\qint_{i,j_1}(\bids_i), \qint_{i,j_2}(\bids_i), \ldots, \qint_{i,j_n}(\bids_i)$ will be (a prefix of) the positional interim allocation $\mathbf{q_i^{\texttt{pos}}}$, which, as per Lemma~\ref{lemma:interimmonotonicity} is monotone, $\qpos_{i,1} \geq \ldots \geq \qpos_{i,m}$. That is, the interim allocations are non-increasing in the reported preference order and are always a permutation of $\mathbf{q_i^{\texttt{pos}}}$. Note that, the expected utility of agent $i$ as per her true preference $\val_i$ is $\sum_{k=1}^m \qint_{i,j_k}(\bids_i) \cdot v_{i, j_k}$, which essentially is a vector dot product of the true valuation and the interim allocation. This dot product is maximized when the true valuations and the interim allocations are \emph{aligned} (i.e., are in the same order). Given that the interim allocation is always a permutation of $\mathbf{q_i^{\texttt{pos}}}$, this alignment happens when agent $i$'s reported preference order matches her true preference order, i.e., when agent $i$ reports $\val_i$, establishing that $\rrpass$ is BIC. This intuition is formalized in the following lemma.

\begin{restatable}{lemma}{RRisBIC}\label{lemma:RR is BIC}
The \textsc{Round-Robin}$^{\text{pass}}$ mechanism is Bayesian incentive compatible for neutral priors.
\end{restatable}
\begin{proof}
Consider any agent $i$ with a valuation $\val_i$ and some report $\bids_i$; let $(\prefer^v_i,m_i^v)$ and $(\prefer^b_i,m_i^b)$ be the tuples corresponding to $\val_i$ and $\bids_i$ respectively. Without loss of generality, we assume that $m_i^b = m$ because reporting $(\prefer^b_i,m)$ instead of $(\prefer^b_i,m_i^b)$ can only only increase the expected utility of agent $i$: since $m \geq m_i^b$~\Cref{lemma:itemindependence} implies that, $\qint_{i,j}(\prefer^b_i,m) \geq \qint_{i,j}(\prefer^b_i,m_i^b)$ for all items $j \in \items$. Thus, to show that $\rrpass$ is BIC, we will show that the expected utility of agent $i$ as per $\val_i$ when reporting $(\prefer^v_i,m_i^v)$ (i.e., $\val_i$) is at least as much as when reporting $(\prefer^b_i,m)$ (i.e., $\bids_i$).

\emph{\bf Case $\mathbf{(1)}$}: $\prefer^v_i = \prefer^b_i$. Using \Cref{lemma:itemindependence}, the expected utility of agent $i$ when reporting $(\prefer^b_i,m)$ is,
\begin{align*}
&\mathop{\E}\limits_{\val_{-i} \sim \mathcal{D}_{-i}} [ u_i((\prefer^b_i,m), \val_{-i}) ] \\
= & \mathop{\E}\limits_{\val_{-i} \sim \mathcal{D}_{-i}} [ u_i((\prefer^v_i,m), \val_{-i}) ] \tag{$\prefer^v_i = \prefer^b_i$} \\
= & \sum_{\ell=1}^m \qpos_{i,\ell} \cdot v_{i,\prefer^v_i(\ell)} \tag{\Cref{lemma:itemindependence}} \\
= & \sum_{\ell=1}^{m_i^v} \qpos_{i,\ell} \cdot v_{i,\prefer^v_i(\ell)} + \sum_{\ell=m_i^v + 1}^m \qpos_{i,\ell} \cdot v_{i,\prefer^v_i(\ell)} \\
= & \sum_{\ell=1}^{m_i^v} \qpos_{i,\ell} \cdot v_{i,\prefer^v_i(\ell)} + \sum_{\ell=m_i^v + 1}^m \qpos_{i,\ell} \cdot 0 \\
= & \mathop{\E}\limits_{\val_{-i} \sim \mathcal{D}_{-i}} [ u_i((\prefer^v_i,m_i^v), \val_{-i}) ] \tag{\Cref{lemma:itemindependence}},
\end{align*}
Hence, in this case, the expected utility of reporting $(\prefer^v_i,m_i^v)$ is at least as much as when reporting $(\prefer^b_i,m)$, which in turn is at least the expected utility when reporting $(\prefer^b_i,m_i^b)$. 

\noindent
\emph{\bf Case $\mathbf{(2)}$}: $\prefer^v_i \neq \prefer^b_i$. Let $k$ be the lowest position where $\prefer^v_i$ and $\prefer^b_i$ differ such that $\prefer^b_i(k) = x$ and $\prefer^v_i(k) = y$. Additionally, let $k'$ be such that $\prefer^b_i(k') = y$; note that $k'> k$ by the definition of $k$.



We will first show that agent $i$ can weakly increase her expected utility by reporting the tuple $(\prefer'_i,m)$, instead of $(\prefer^b_i,m)$, where ordering $\prefer'_i$ is constructed by swapping the position of $x$ and $y$ in $\prefer^b_i$. We will then apply this swapping argument repeatedly to conclude the proof. Towards this, note that $v_{i,y} \geq v_{i,x}$ since the preference order $\prefer_i$ is based on the true valuations, and additionally by~\Cref{lemma:interimmonotonicity}, we have $\qpos_{i,k} \geq \qpos_{i,k'}$. Together, these two inequalities imply that $\qpos_{i,k} \cdot v_{i,y} + \qpos_{i,k'} \cdot v_{i,x} \geq \qpos_{i,k} \cdot v_{i,x} + \qpos_{i,k'} \cdot v_{i,y}$. Given this, the expected utility of reporting $(\prefer'_i,m)$ can be written as:
\begin{align*}
    & \mathop{\E}\limits_{\val_{-i} \sim \mathcal{D}_{-i}} [ u_i((\prefer'_i,m), \val_{-i}) ]\\
    = & \
    \qpos_{i,k} \cdot v_{i,y} + \qpos_{i,k'} \cdot v_{i,x} + \sum_{\ell \notin \{k,k'\}} \qpos_{i,\ell} \cdot v_{i,\prefer^b_i(\ell)} \tag{\Cref{lemma:itemindependence}}\\
    \geq & \ \qpos_{i,k} \cdot v_{i,x} + \qpos_{i,k'} \cdot v_{i,y} + \sum_{\ell \notin \{k,k'\}} \qpos_{i,\ell} \cdot v_{i,\prefer^b_i(\ell)}\\
  =& \sum_{\ell = 1}^m \qpos_{i,\ell} \cdot v_{i,\prefer^b_i(\ell)} \tag{$\prefer^b_i(k) = x$ and $\prefer^b_i(k') = y$}\\
  =& \mathop{\E}\limits_{\val_{-i} \sim \mathcal{D}_{-i}} [ u_i((\prefer^b_i,m), \val_{-i}) ] \tag{\Cref{lemma:itemindependence}}.
\end{align*}

Thus, this ``swapping'' step used to obtain $\prefer'_i$ from $\prefer^b_i$ weakly increases the expected utility. Additionally, by definition, the length of the longest common prefix of $\prefer'_i$ and $\prefer_i^v$ is strictly greater than that of $\prefer^b_i$ and $\prefer_i^v$, i.e., $\prefer'_i$ is closer to $\prefer_i^v$ than $\prefer^b_i$. Therefore, by repeatedly performing these swaps till $\prefer'_i$ becomes $\prefer_i^v$, we get that the expected utility of reporting $(\prefer_i^v, m)$ is at least as when reporting $(\prefer^b_i,m)$. Following the arguments of Case $(1)$, we also know that the expected utility of reporting $(\prefer_i^v, m^v_i)$ is at least as when reporting $(\prefer^v_i,m)$. Hence, in this case, we also get that the expected utility of reporting $(\prefer^v_i,m_i^v)$ is at least as much as when reporting $(\prefer^b_i,m)$. That is, reporting $\val_i$ is always better than reporting any other $\bids_i$. This concludes the proof.
\end{proof}

The following lemma proves that $\rrpass$ also satisfies SD$^+$efficiency and $\EFone$, its proof can be found in~\Cref{appendix:rr missing proofs}.
\begin{restatable}{lemma}{RRsdpoplus}\label{lemma: rrsdpoplus}
    The \textsc{Round-Robin}$^{\text{pass}}$ mechanism outputs SD$^+$efficient and $\EFone$ allocations.
\end{restatable}
Combined, Lemmas~\ref{lemma:RR is BIC} and~\ref{lemma: rrsdpoplus} imply~\Cref{thm: rr main}.

\subsection{Welfare functions and BIC}\label{subsec:pigou dalton negative}

In this section, we prove that a large family of ``welfare maximization'' mechanisms is not Bayesian incentive compatible.

A \emph{welfare function} $f : \mathbb{R}^n_{\geq 0} \to \mathbb{R}$ is a function that, given the utilities of agents $u_1(\allo), \ldots, u_n(\allo)$ in allocation $\allo$, outputs $f(u_1(\allo), \ldots, u_n(\allo))$, that represents the collective welfare of all the agents in allocation $\allo$. An allocation $\allo$ is said to maximize $f$ if for all other feasible allocations $\alloy$ we have $f(u_1(\allo), \ldots, u_n(\allo)) \geq f(u_1(\alloy), \ldots, u_n(\alloy))$.

\begin{definition}[Anonymous]
    A welfare function $f$ is called \emph{anonymous} iff $f(x_1, x_2, \ldots, x_n) = f(x_{\pi_1}, x_{\pi_2}, \ldots, x_{\pi_n})$ for any permutation $\pi : [n] \to [n]$.
\end{definition} 

\begin{definition}[Strictly monotone]
    A welfare function $f$ is \emph{strictly monotone} iff $f(x_1, x_2, \ldots, x_n) > f(y_1, y_2, \ldots, y_n)$ whenever we have $x_i \geq y_i$ for all $i \in [n]$ and $x_j > y_j$ for some $j \in [n]$.
\end{definition}

\begin{definition}[The Pigou-Dalton principle]
    The Pigou-Dalton Principle (PDP) states that a welfare function should weakly prefer a profile that reduces the inequality between two agents, assuming that the utilities of all other agents stay unchanged. Formally, a welfare function $f$ satisfies PDP iff for any two vectors $(x_1, \ldots, x_n)$ and $(y_1, \ldots, y_n)$ such that there are two distinct $i, j \in [n]$ satisfying $(i)$ $x_i + x_j = y_i + y_j$, $(ii)$ $x_i > x_j$,  $(iii)$ $x_i > y_i > x_j$ (equivalently $x_i > y_j > x_j$), and $(iv)$ $x_k = y_k$ for all $k \in [n] \setminus \{i,j\}$, we have $f(x_1, x_2, \ldots, x_n) \leq f(y_1, y_2, \ldots, y_n)$.
\end{definition}

Welfare functions satisfying anonymity, strict monotonicity, and the PDP constitute a broad set of functions that include, e.g., the class of all the $p$-mean welfare functions\footnote{For a given $p \in \mathbb{R}_{\geq 0}$, the $p$-mean welfare function 
is defined as $f_p(x_1, x_, \ldots, x_n) \coloneqq  \big( \frac{1}{n} \sum_{i=1}^n x_i^p \big)^{1/p}$. Notably, $f_p$ captures utilitarian welfare if $p=1$, Nash welfare if $p \to 0$, and egalitarian welfare if $p \to -\infty$.} for $p \leq 1$~\cite{moulin2004fair}.  The class of $p$-mean welfare functions, in turn, includes Nash social welfare, egalitarian social welfare, utilitarian social welfare, and the leximin criterion. In Theorem \ref{theorem:BIC-welfare-max}, we show that, for any deterministic (or randomized) mechanism that outputs allocations that (ex-post) maximize such a welfare function, $f$, is not BIC. We prove this fact for a very simple neutral prior distribution according to which each agent $i$'s values are drawn uniformly from the $m-1$-simplex, i.e., $\mathbf{v}_i = (v_1, v_2, \ldots, v_m) \sim \Delta^{m-1}$. Notably, Theorem \ref{thm: rr main} also holds for this neutral prior supported over normalized valuations.\footnote{It is necessary to consider normalized valuations when studying welfare objectives that aren't scale-free, for e.g. social welfare. In particular, if the valuations aren't normalized, then agents could (mis)report large values for all the items, thereby forcing the mechanism to allocate all the items to them.}



\begin{restatable}{theorem}{BICWelfare}\label{theorem:BIC-welfare-max}
    There is no deterministic (or randomized) mechanism that is Bayesian incentive compatible for the uniform prior and which always outputs allocations that ex-post maximize welfare function $f$, where $f$ is any anonymous, strictly monotone function satisfying the Pigou-Dalton principle.
\end{restatable}
\begin{proof}
    We will show that this negative result holds even for the simple setting of $n=2$ agents and $m=2$ items. Let $\allo$ be any deterministic or randomized BIC mechanism that (ex-post) outputs allocations that maximize $f$. To arrive at a contradiction, we will show that show that if agent $1$ has a true normalized preference $\val_1 = (x,1-x)$ such that $x > 1/2$, then it is better for the agent to report $\bids_1 = (x+ \epsilon, 1 - x - \epsilon)$ for any $\epsilon \in (0, 1-x]$. This would show that $\allo$ is not BIC since truthful reporting is not the best strategy for agent $1$.
    More specifically, we will show that if  agent $1$ having true preference $\val_1 = (x, 1-x)$ reports $\bids_1 = (b, 1-b)$,where $b>x$, then her expected utility $\E_{\val_{2} \sim \Delta_2}[u_1(\allo(\bids_1, \val_2)] = xb + (1-x)(1-b)$ for any mechanism $\allo$. Note that this will imply that $\E_{\val_{2} \sim \Delta_2}[u_1(\allo(\bids_1, \val_2)] > \E_{\val_{2} \sim \Delta_2}[u_1(\allo(\val_1, \val_2)]$ since

    \begin{align*}
    & \mathop{\E}_{\val_{2} \sim \Delta_2}[u_1(\allo(\bids_1, \val_2)] - \mathop{\E}_{\val_{2} \sim \Delta_2}[u_1(\allo(\val_1, \val_2)]\\
    & = \big( x(x+\epsilon) + (1-x)(1-x - \epsilon) \big) - \big( x^2 + (1-x)^2 \big)\\
    & = \epsilon x - \epsilon (1-x) = \epsilon (2x-1) > 0,  \numberthis \label{ineq:bic-welfare}
    \end{align*}
    
     the final inequality follows from the fact that $x>1/2$ and $\epsilon > 0$. This will imply that $\allo$ is not BIC, giving us the desired contradiction. To conclude the proof, we will show that $\E_{\val_{2} \sim \Delta_2}[u_1(\allo(\bids_1, \val_2)] = xb + (1-x)(1-b)$. Consider $\val_2 = (y, 1-y) \sim \Delta_2$. First, will show that if $y > b$, then the unique allocation that maximizes $f$ for the report $(\bids_1, \val_1)$ is $\{\{2\}, \{1\} \}$. Otherwise, if $y < b$, then the unique allocation that maximizes $f$ is $ \{\{1\}, \{2\}\}$.
     Suppose $y > b$. Since there are two items and two agents, there are four possible allocations having the utility profiles $P = \{ (1, 0), (0, 1), (b, 1-y), (1-b, y)\}$. Next we will show that $f(1-b,y) > f(p,q)$, for any $(p,q) \in P \setminus \{(1-b,y)\}$, implying the allocation $\allo = \{\{2\}, \{1\} \}$ is the unique allocation that maximizes $f$. Towards this, note that $f(1-b,y) = f(y, 1-b) < f(b, 1-y)$; the first equality follows from the fact that $f$ is symmetric and second from the fact that $f$ is strictly monotone and $b>y$ (thereby $1-y > 1-b$). Additionally, $f(1,0) = f(0,1) < f(b, 1-b) < f(b, 1-y)$; here the first, second, and the third step follow respectively from the fact that $f$ is symmetric, satisfies PDP, and strictly monotone along with the fact that $1-y > 1-b$. 
     
     Thus, if $y>b$ we have $f(1-b,y) > f(p,q)$, for any $(p,q) \in P \setminus \{(1-b,y)\}$, which implies that the allocation $\allo = \{\{2\}, \{1\} \}$ is the unique allocation that maximizes $f$. Note that using the same argument, we can get that if $y<b$, then the unique allocation that maximizes $f$ is $\allo = \{\{1\}, \{2\} \}$.
     
     Finally, we will use this to prove that $\E_{\val_{2} \sim \Delta_2}[u_1(\allo(\bids_1, \val_2))] = xb + (1-x)(1-b)$. Note that if $\val_2 = (y, 1-y) \sim \Delta_2$ then the distribution of $y$ is uniform, i.e., $y \in \mathcal{U}[0,1]$. Therefore, $y < b$ with probability $b$ and $y > b$ with probability $1-b$. Combining this with the aforementioned observation we have that with probability $b$ the unique $f$ maximizing allocation is $\allo = \{\{1\}, \{2\} \}$ wherein agent $1$'s true utility is $x$, and with probability $1-b$ the unique $f$ maximizing allocation is $\allo = \{\{2\}, \{1\} \}$ wherein agent $1$'s true utility is $1-x$. Since these allocations are unique, any randomized or deterministic mechanism $\allo$ must (ex-post) output them. Therefore, $\E_{\val_{2} \sim \Delta_2}[u_1(\allo(\bids_1, \val_2)] = xb + (1-x)(1-b)$, this concludes the proof.
\end{proof}



So far we have been focusing on deterministic allocations, i.e., $x_{i,j} \in \{0,1\}$ for all $i$, $j$. A \emph{randomized} allocation is a distribution over a set of deterministic allocations, where $x_{i,j} \in [0,1]$ denotes the probability item $j$ is assigned to agent $i$ received under the allocation. We say that an allocation is \emph{ex-post fPO} if no randomized allocation Pareto dominates it. A slightly modified version of the above proof can be used to show the non-existence of deterministic (or randomized) mechanisms that are Bayesian incentive compatible and output allocations that are ex-post $\fPO$ and EF1. In fact, we can rule out BIC mechanisms that output $\fPO$ allocations that satisfy a fairness criterion we call \emph{fulfillment}, that is much weaker than $\EFone$. Specifically, an allocation $\allo$ is fulfilling iff the following condition is satisfied for all agents $i \in \agents$: if $i$ has at least $n$ items that she values positively, then agent $i$ gets a bundle of positive utility. It is easy to see that an $\EFone$ allocation must be fulfilling. Formally, this result is stated as Theorem \ref{thm: fulfilling}.


\begin{restatable}{theorem}{BICfulfilling}\label{thm: fulfilling}
    There is no deterministic (or randomized) Bayesian incentive compatible mechanism that outputs $\fPO$ allocations that are fulfilling.
\end{restatable}
\begin{proof}
    Structurally, the proof is similar to the proof of Theorem \ref{theorem:BIC-welfare-max}. For the sake of brevity, we will only describe the points where the proof differs. In proof of Theorem \ref{theorem:BIC-welfare-max}, we showed that allocations that maximize $f$ are unique if $y>b$ or $y<b$. Similarly, here we will show that if $y<b$, then $\allo = \{\{1\}, \{2\} \}$ is the unique allocation that is $\fPO$ and fulfilling, and if $y> b$, then $\allo = \{\{2\}, \{1\} \}$ is the unique such allocation. Proving this would imply that for any randomized or deterministic mechanism $\allo$ that outputs $\fPO$ and fulfilling allocations, $\E_{\val_{2} \sim \Delta_2}[u_1(\allo(\bids_1, \val_2)] = xb + (1-x)(1-b)$ where the bid and the true value of agent $1$ is $\bids_{1} = (b, 1-b)$ and $\val_1 = (v, 1-v)$ respectively. Finally, together with inequality (\ref{ineq:bic-welfare}) this will imply that $\allo$ is not BIC and would conclude the proof.
    
    Now consider the case when $y < b$; since $y \sim \mathcal{U}[0,1]$ almost surely $y>0$, thus, $0 < y < b$. Any fulfilling allocation for this instance must give one item to each of the agents since both agents value $n=2$ items positively. These allocations have utility profiles $(b, 1-y)$ and $(1-b,y)$. Among these two profiles, $(1-b, y)$ is not $\fPO$ because we can show that it is Pareto dominated by a linear combination of the profiles $(0,1)$ and $(b,1-y)$. Specifically, consider the following linear combination, $(a,b) = (1-\alpha)(0,1) + \alpha(b,1-y)$ where $\alpha = \frac{1-y}{y}$. On simplifying we get, $(a,b) = (\frac{b(1-y)}{y}, 1- \alpha y) = (\frac{b(1-y)}{y},y)$. The profile $(a,b)$ Pareto dominates $(1-b, y)$ because $b = y$ and $a = \frac{b(1-y)}{y} > 1-b$ where the final inequality follows from the fact that $b>y$. Therefore, if $y<b$ then allocation $\allo = \{\{1\}, \{2\} \}$ having utility profile $(b, 1-y)$ is the only $\fPO$ and fulfilling allocation.

    If $y>b$, then we can show that the profile $(b,1-y)$ is not $\fPO$, thereby making $(1-b,y)$ the only $\fPO$ and fulfilling profile. Towards this consider the following linear combination of profiles: $(a',b') = \beta (1-b,y) + (1-\beta)(1,0)$ for $\beta = \frac{1-y}{y}$. On simplifying we get, $(a',b') = (1 - \beta b, 1-y) = (\frac{y-b+by}{y},1-y)$; this Pareto dominates $(b,1-y)$ because $\frac{y-b+by}{y} > b$ which is implied by $y>b$. Hence, if $y>b$, then the allocation $\allo = \{\{2\}, \{1\} \}$ having utility profile $(1-b,y)$ is the only $\fPO$ and fulfilling allocation. This concludes the proof.
\end{proof}

\section{Allocating Divisible Goods: Cake Cutting}\label{section:cake-cutting}

Recently,~\cite{tao2022existence} proved a strong incompatibility between fairness and incentives in the cake-cutting problem: no deterministic DSIC cake-cutting mechanism can always output proportional allocations (see Theorem \ref{theorem:cake_cutting_neg}), for the case of $n=2$ agents and piecewise constant density functions.

\begin{theorem}[\cite{tao2022existence}]
There is no deterministic, DSIC, and proportional cake-cutting mechanism for piecewise-constant valuations, even for $n=2$ agents.\label{theorem:cake_cutting_neg}
\end{theorem}

\paragraph{The Bayesian Setting.} The \emph{private} piecewise-constant valuation function $\ft_i$ of each agent $i$ is drawn from a \emph{neutral} prior distribution $\matD_i$. We say a prior distribution $\matD$ is neutral, iff for all integers $k>1$, and disjoint pieces of cake $X_1, X_2, \ldots, X_k$ satisfying $|X_1| = |X_2| = \ldots = |X_k|$, we have
$\mathop{\Pr}\limits_{\ft \sim \matD}[ j = \arg\max_{i=1}^k \ft(X_i)] = 1/k.$

Intuitively, neutral priors represent a distribution of density functions such that the probability of an agent preferring an interval over another same-length interval is the same, i.e., same-length intervals are \emph{equivalent} in expectation. 


Towards bypassing Theorem \ref{theorem:cake_cutting_neg}, we provide a deterministic mechanism (Mechanism \ref{mec:IA}), which we refer to as \ia\ 
(IA), and we show that 
IA is Bayesian incentive compatible for neutral priors and it always outputs proportional allocations (Theorem \ref{thm: positive for cake cutting}). Moreover, IA can be executed in polynomial time for piecewise constant density functions.

\begin{theorem}\label{thm: positive for cake cutting}
       The \ia\ mechanism is Bayesian incentive compatible for neutral product distributions, and always outputs proportional allocations.
\end{theorem}


At a high level, given reported preferences $\fr_1, \fr_2, \ldots, \fr_n$, IA operates as follows. Initially, the entire cake $[0,1]$ is allocated to agent $1$. Other agents arrive one-by-one following the order $2, 3, \ldots, n$. Before the arrival of agent $i$, the entire cake is always allocated among the agents $\{1, 2, \ldots, i-1\}$; this allocation is denoted by $X_1, X_2, \ldots, X_{i-1}$. When agent $i$ arrives, each agent $j \in \{1, 2, \ldots, i-1\}$ cuts their piece of cake, $X_j$, into $i$ pieces, $C_j^1, C_j^2, \ldots, C_j^i$ using the \se
operation (Subroutine \ref{mec:SE}) described below.

\noindent
{\bf SplitEqual operation.} Given a piece of cake $X$, a function $f$, and a number $k$, the \se\ operation partitions $X$ into pieces $C^1, C^2, \ldots, C^k$ having \emph{equal sizes} and \emph{equal values}, i.e., $|C^1| = |C^2| = \ldots = |C^k|$ and $f(C^1) = f(C^2) = \ldots = f(C^k)$. Note that this operation can always be performed given ``any'' density function $f$ and $k>1$ since the existence of such a split is directly implied by the existence of a \emph{perfect partition}.\footnote{\cite{alon1987splitting} proved that given arbitrary density functions $y_1, y_2, \ldots, y_n$, there always exists a partition $X_1, X_2, \ldots, X_n$ of $[0,1]$ such that $y_i(X_j) = 1/n$ for all $i,j \in [n]$.} Specifically, if we consider $k$ density functions where the first function, $y_1(X) = |X|$ while the other $k-1$ density functions are $f$, then a perfect partition $C^1, C^2, \ldots, C^i$ of $X$ --- which always exists --- will satisfy the two required conditions of equal value (since $f(C^r) = f(C^s)$ for all $r,s \in [i]$) and equal size (since $|C^r| = y_1(C^r) = y_1(C^s) = |C^s|$ for all $r,s \in [k]$). The \se\ operation can be efficiently performed for piecewise constant valuations by simply dividing each subinterval within $X$ having a constant-value w.r.t.\ $f$ into $k$ pieces of equal size; the equal value condition will follow from the fact the subinterval being divided has a constant-value. Furthermore, it is known that a perfect partition can be computed efficiently for the significantly more general class of piecewise linear valuations~\cite{chen2013truth}, allowing for efficient execution of IA for this class of valuations as well.

After performing the \se\ operation, each agent $j \in \{ 1, \dots, i-1 \}$ offers agent $i$ one piece amongst $C_j^1, C_j^2, \ldots, C_j^i$, and agent $i$ picks the one that maximizes her value, as per the reported density function $\fr_i$. The resulting allocation, after the arrival of the $n^{\text{th}}$ agent, is then returned as the final allocation. In Lemma \ref{theorem:cake_bic}, we show that IA is BIC for neutral priors. In Lemma~\ref{theorem:cake_prop}, we show that IA always outputs proportional allocations. The proof of  Lemma~\ref{theorem:cake_prop} resembles a proof of~\cite{fink1964note} wherein proportionality is established for a similar recursive cake-cutting algorithm.

\begin{algorithm}[H]
\textbf{Input:} $(f_1, \dots, f_n)$, the reported valuations of the agents\\
\textbf{set} $X_1 \gets [0,1]$, and $X_i \gets \emptyset$ for all $i \in \{2,\dots, n\}$ \\
\For{ $i \gets 2, 3, \dots, n$}{
    \For{$j \gets 1, 2, \dots,i-1$}{
        \textbf{set} $(C^1_j, C^2_j, \dots, C^{i}_j) \gets $ \se($X_j, f_j, i$)\\
        \textbf{let} $k^* \gets  \argmax_k f_i(C_j^k)$\\
        \textbf{set} $X_i \gets X_i \cup C_j^{k^*}$\\
        \textbf{set} $X_j \gets X_j \setminus C_j^{k^*}$
    }
}
\Return{$(X_1, X_2, \dots, X_n)$}
\caption{\ia}
\label{mec:IA}
\end{algorithm}

\renewcommand{\algorithmcfname}{Subroutine}

\begin{algorithm}[H]
\textbf{Input:} $(X, f, k)$ a set of intervals $X$, valuation function $f$, and number $k > 1$\\
Partition $X$ into $k$ subintervals $C^1, C^2 \dots, C^k$ such that the following two conditions are true:
\begin{enumerate}
    \item [(1)] $|C^i| = |C^j|$ for all $i,j \in [k]$
    \item [(2)] $f(C_i) = f(C_j)$ for all $i,j \in [k]$
\end{enumerate}
\Return{$(C^1, C^2, \dots, C^k)$}
\caption{\se}
\label{mec:SE}
\end{algorithm}
\renewcommand{\algorithmcfname}{MECHANISM}

\begin{restatable}{lemma}{expectedUtilLemma}\label{lem:euqalshare}
Let $X'_j$, $X_j$ be the pieces of cake that agent $j$ is allocated at the beginning and the end of iteration $t$, respectively, for some $1 \leq j < t \leq n$.\footnote{Noting that $X_j$ is a random variable which depends on the valuation function of agent $t$, $\ft_t$, drawn from $\mathcal{D}$.} Then, for any valuation function $\hat{f}$, the expected value of $\hat{f}(X_j)$ is equal to $\frac{t-1}{t} \cdot \hat{f}(X'_j)$, where the expectation is over the randomness of the valuation function of agent $t$, $\ft_t$, drawn from a neutral prior $\mathcal{D}$. 
Formally,
\[\mathop{\E}\limits_{\ft_t \sim \mathcal{D}}[\hat{f}(X_j)] = \frac{t-1}{t} \cdot \hat{f}(X'_j).\]
\end{restatable}
\begin{proof}
In iteration $t$, via the split-equal subroutine, the piece $X'_j$ is partitioned into disjoint intervals $C_j^1, C_j^2, \dots, C_j^t$. Then agent $t$, having a valuation function $\ft_t \sim \mathcal{D}_t$ picks her favorite piece amongst $C_j^1, C_j^2, \ldots, C_j^t$. For each $k \in [t]$, define $a_k$ to be the indicator random variable such that $a_k = 1$ if the piece $C_j^k$ is picked by agent $t$, and $0$ otherwise. If agent $t$ picks the piece $C_j^k$ in iteration $t$, then the piece of cake agent $j$ has at the end of iteration satisfies $X_j = X'_j \setminus C_j^k$. This allows us to express the expected value $\E_{\ft_t \sim \mathcal{D}_t}[\hat{f}(X_j)]$ as, 
\begin{align*}
    \mathop{\E}\limits_{\ft_t \sim \mathcal{D}_t}[\hat{f}(X_j)] &= \mathop{\E}\limits_{\ft_t \sim \mathcal{D}_t}\left[ \sum_{k=1}^t a_k \Big( \hat{f}(X'_j) - \hat{f}(C_j^k) \Big) \right]\\
    & = \hat{f}(X'_j) \mathop{\E}\limits_{\ft_t \sim \mathcal{D}_t}\left[ \sum_{k=1}^t a_k \right] - \mathop{\E}\limits_{\ft_t \sim \mathcal{D}_t} \left[\sum_{k=1}^t a_k \hat{f}(C_j^k)\right],
\end{align*}
where the final inequality follows from the fact that $\hat{f}(X'_j)$ and $\hat{f}(C_j^k)$ are independent of $f_t$. To obtain the desired equality we will now use the following two observations. First, the prior $\mathcal{D}_t$ is neutral, i.e., $\E_{f_t \sim \mathcal{D}_t} [X_k] = 1/t$ because the $t$ intervals $C_j^1, C_j^2, \dots, C_j^t$ have equal sizes, $|C_j^1| = |C_j^2| = \dots = |C_j^t|$. Second, $\sum_{k=1}^t X_k = 1$ since agent $t$ always picks a piece from agent $j$. Using these two observations we get,

\begin{align*}
    \mathop{\E}\limits_{f_t \sim \mathcal{D}_t}[\hat{f}(X_j)] &= \hat{f}(X'_j) \mathop{\E}\limits_{\ft_t \sim \mathcal{D}_t}\left[ \sum_{k=1}^t a_k \right] - \mathop{\E}\limits_{\ft_t \sim \mathcal{D}_t} \left[\sum_{k=1}^t a_k \hat{f}(C_j^k)\right]\\
    & = \hat{f}(X'_j) \cdot 1 - \sum_{k=1}^t \frac{1}{t} \hat{f}(C_j^k) \\
    & = \hat{f}(X'_j) - \frac{1}{t} \hat{f}(X'_j) \tag{$\cup_{k=1}^t C_j^k = X'_j$} \\
    & = \frac{t-1}{t} \cdot \hat{f}(X'_j).\qedhere
\end{align*}

\end{proof}

Now using Lemma \ref{lem:euqalshare}, we will show that IA is Bayesian incentive compatible for neutral priors. At a high level, we show that if agent $i$ has piece $X^i_i$ after her arrival, then, eventually, in the final allocation, her utility, as per her true valuation $\ft_i$, will be proportional to $\ft_i(X^i_i)$. Specifically, let $X^n_i$ be the piece that agent $i$ has at the end of the mechanism's execution. Through repeated applications of Lemma~\ref{lem:euqalshare} we prove that $\ft_i(X_i^n) = \frac{i}{n}\cdot \ft_i(X^i_i)$ 
for any reported valuation function $\fr_i$. Therefore, it suffices to prove that truthfully reporting $\ft_i$ maximizes $\ft_i(X^i_i)$. 
\begin{restatable}{lemma}{cakecuttingBICLemma}\label{theorem:cake_bic}
      The \ia\ is Bayesian incentive compatible with respect to product distributions that are neutral.  
\end{restatable}
\begin{proof}
Fixing an agent $i < n$, we will show that IA is BIC wrt neutral prior $\times_{j=1}^n \mathcal{D}_j$ for agent $i$; for agent $n$, the mechanism is trivially DSIC. Let the true valuation of agent $i$ be $\ft_i$ and let her reported valuation be $\fr_i$. Denote by $X_i^t$ the piece of cake agent $i$ has at the end of iteration $t$; for $t=1$, define $X_i^1 = [0,1]$ for $i=1$, and $X_i^1 = \emptyset$ otherwise. By definition of our mechanism, we have $X_i^t = \emptyset$ for $t<i$. Additionally, note that the reported function $\fr_i$ has no influence on the execution of the mechanism prior to the arrival of agent $i$: agent $i$, using her report $\fr_i$, cannot change the allocation $X_1^{i-1}, X_2^{i-1}, \ldots, X_n^{i-1}$. 

In the subsequent analysis, we will show that reporting $\ft_i$ is at least as good as reporting any $\fr_i$ for agent $i$ even if we fix the valuations of agents in $[i-1]$, i.e., IA is BIC for agent $i$ conditioned on valuations $\fr_1, \fr_2, \ldots, \fr_{i-1}$. Formally, for any fixed  true valuation functions $\ft_1, \ft_2, \ldots, \ft_{i-1}$,
\begin{align}\label{ineq:bic_main}
    \mathop{\E}\limits_{\ft_{i+1}, \ft_{i+2}, \ldots, \ft_{n}}[u_i(\text{IA}(\ft_i, \mathbf{f}^*_{-i})] \geq \mathop{\E}\limits_{\ft_{i+1}, \ft_{i+2}, \ldots, \ft_{n}}[u_i(\text{IA}(\fr_i, \mathbf{f}^*_{-i})].
\end{align}
Note that the above inequality implies that IA is BIC for agent $i$ because taking an expectation over $\ft_1 \sim \mathcal{D}_1, \ft_2 \sim \mathcal{D}_2, \ldots, \ft_{i-1} \sim \mathcal{D}_{i-1}$ on both sides of the above inequality gives us
\begin{align*}
    \mathop{\E}\limits_{\mathbf{f}^*_{-i} \sim \mathcal{D}_{-i}}[u_i(\text{IA}(\ft_i, \mathbf{f}^*_{-i})] \geq \mathop{\E}\limits_{\mathbf{f}^*_{-i} \sim \mathcal{D}_{-i}}[u_i(\text{IA}(\fr_i, \mathbf{f}^*_{-i})].
\end{align*}
To complete the proof, we will now prove inequality (\ref{ineq:bic_main}). First, we will show that, for any fixed $\ft_1, \ft_2, \ldots, \ft_{i-1}$ and report of agent $i$, $\fr_i$, we will have
\begin{align}\label{ineq:bic_expectation}
    \mathop{\E}\limits_{\ft_{i+1}, \ft_{i+2}, \ldots, \ft_{n}}[\ft_i(X_i^n)] = \frac{i}{n}\cdot \ft_i(X^i_i).
\end{align}
Consider the following equality, $\mathop{\E}\limits_{\ft_n}[\ft_i(X_i^n)] = \frac{n-1}{n} \ft_i(X_i^{n-1})$, implied by Lemma~\ref{lem:euqalshare}, 

To prove equation (\ref{ineq:bic_expectation}), starting from the above equality, we will repeatedly take expectation over $\ft_j \sim \mathcal{D}_j$ over both sides and we will invoke Lemma \ref{lem:euqalshare} for agent $i$ and iteration $j$, beginning from $j = n-1$ down to $j=i+1$. Performing this sequence of operations gives us the following chain of equalities, leading to equation (\ref{ineq:bic_expectation}),
\begin{align*}
    \mathop{\E}\limits_{\ft_n \sim \mathcal{D}_n}& [\ft_i(X_i^n)] = \frac{n-1}{n} \ft_i(X_i^{n-1}) \\
    \implies  \mathop{\E}\limits_{\ft_{n-1}, \ft_n}&[\ft_i(X_i^n)] =  \frac{n-1}{n}\mathop{\E}\limits_{\ft_{n-1}}[\ft_i(X_{i}^{n-1})] = \frac{n-1}{n} \cdot\frac{n-2}{n-1} \ft_i(X_i^{n-2}) \tag{via Lemma \ref{lem:euqalshare}}\\
    & \; \dots\\
    \implies \mathop{\E}\limits_{\ft_{i+1}, \ldots, \ft_{n}}& [\ft_i(X_i^n)] = \frac{n-1}{n} \cdot \frac{n-2}{n-1} \cdots \frac{i-1}{i}\ft_i(X^i_i) = \frac{i}{n}\ft_i(X^i_i).\tag{via Lemma \ref{lem:euqalshare}}
\end{align*}
Finally, using equation (\ref{ineq:bic_expectation}) we will prove inequality (\ref{ineq:bic_main}). For fixed $f_1, f_2, \ldots, f_{i-1}$, if agent $i$ reports $f_i$, then the piece $X^i_i = \cup_{j=1}^{i-1} C_j^{k_j}$ where $C_j^{k_j}$ is the piece that agent $i$ picked from agent $j$, i.e., we have $k_j \in \arg \max_{k=1}^i f_i(C_j^k)$. Thus,the mechanism picks the pieces $X^i_i$ that maximize value wrt the reported value $\fr_i$. However, given the pieces picked by her, equality (\ref{ineq:bic_expectation}) implies that the final expected utility she derives is directly proportional to $\ft_i(X^i_i)$, her true value for $X_i^i$. Therefore, the final expected utility of agent $i$ is maximized when $\fr_i = \ft_i$,i.e., agent $i$ reports truthfully. We therefore have: \[\mathop{\E}\limits_{\ft_{i+1}, \ft_{i+2}, \ldots, \ft_{n}}[u_i(\text{IA}(\ft_i, \mathbf{f}^*_{-i})] \geq \mathop{\E}\limits_{\ft_{i+1}, \ft_{i+2}, \ldots, \ft_{n}}[u_i(\text{IA}(\ft_i, \mathbf{f}^*_{-i})].\qedhere\]\end{proof}

The following lemma proves that IA outputs proportional allocations, its proof appear in Appendix \ref{appendix:cake cutting proofs}.

\begin{restatable}{lemma}{IAprop}\label{theorem:cake_prop}
    The \ia\ mechanism outputs proportional allocations.
\end{restatable}

Combining Lemmas~\ref{theorem:cake_bic} and~\ref{theorem:cake_prop} implies~\Cref{thm: positive for cake cutting}.


\section*{Acknowledgements}
Alexandros Psomas and Paritosh Verma are supported in part by the NSF CAREER award CCF-2144208, a Google Research Scholar Award, and a Google AI for Social Good award. Vasilis Gkatzelis and Xizhi Tan are partially supported by the NSF CAREER award CCF-2047907.

\bibliographystyle{alpha}
\bibliography{biblio.bib}

\appendix


\section{Missing proofs from Section~\ref{subsec:sdpoplus negative result}}\label{appendix:sdpoplus missing proofs}

\sdUtility*
\begin{proof}

    Without loss of generality assume that the items are indexed such that $v_{i,1} \geq v_{i,2} \geq \ldots \geq v_{i,m}$. Let $\ell$ be the largest index of an item for which agent $i$ has a positive value, i.e., $\ell =  \max\limits_{v_{i,j} > 0} j$. Note that $\ell$ is not defined if $v_{i,j} = 0$ for all $j \in [m]$. However, in this case, $\alloa \sdgeq_i \allob$ trivially holds and $\alloa \sdpgt_i \allob$ can never be true. Hence, we will assume that $\ell$ is well-defined. 
    
    \noindent
    \emph{Proving $(1)$.} To establish that $u_i(\alloa) \geq u_i(\allob)$ if $\alloa \sdpgeq_i \allob$, we will inductively show that for each $j \in [\ell]$, the inequality

    \begin{align*}
    \sum_{k=1}^j (y_{i,k} - x_{i,k}) v_{i,k} \geq v_{i,j} \sum_{k=1}^j (y_{i,k} - x_{i,k}), \numberthis \label{equation:prop-sd-utility0}
    \end{align*}

    is satisfied. Note that proving Equation \ref{equation:prop-sd-utility0} for all $j \in [\ell]$ establishes the proposition, since instantiating Equation \ref{equation:prop-sd-utility0} with $j = \ell$ gives us 

    \begin{align*}
    \sum_{k=1}^\ell (y_{i,k} - x_{i,k}) v_{i,k} \geq v_{i,\ell} \sum_{k=1}^\ell (y_{i,k} - x_{i,k}). \numberthis \label{equation:prop-sd-utility1}
    \end{align*}
    The left-hand side of the above inequality is equal to $\sum_{k=1}^\ell (y_{i,k} - x_{i,k}) v_{i,k} = u_i(\alloa) - u_i(\allob)$, and the right-hand side, $v_{i,\ell} \sum_{k=1}^\ell (y_{i,k} - x_{i,k}) \geq 0$ because $v_{i,\ell} > 0$ along with the fact that $\sum_{k=1}^\ell (y_{i,k} - x_{i,k}) = \sum_{k=1}^\ell y_{i,k} - \sum_{k=1}^\ell x_{i,k} \geq 0$; here the final inequality follows from the fact that $\alloa \sdpgeq_i \allob$.

    Now, by induction on $j$, we will show that Equation \ref{equation:prop-sd-utility0} holds for all $j \in [\ell]$. For the base case of the induction, note that Equation \ref{equation:prop-sd-utility0} trivially holds for $j=1$. As part of the induction step, we will show that assuming Equation \ref{equation:prop-sd-utility0} holds for some $j < \ell$ implies that Equation \ref{equation:prop-sd-utility0} holds for $j+1$ as well. To this end, consider the following

    \begin{align*}
        \sum_{k=1}^{j+1} (y_{i,k} - x_{i,k}) v_{i,k} & = \sum_{k=1}^j (y_{i,k} - x_{i,k}) v_{i,k} + (y_{i,j+1} - x_{i,j+1}) v_{i,j+1} \\
        & \geq v_{i,j} \sum_{k=1}^j (y_{i,k} - x_{i,k}) + (y_{i,j+1} - x_{i,j+1}) v_{i,j+1} \tag{Equation (\ref{equation:prop-sd-utility0}) holds for $j$}\\
        & \geq v_{i,j+1} \sum_{k=1}^j (y_{i,k} - x_{i,k}) + (y_{i,j+1} - x_{i,j+1}) v_{i,j+1} \tag{$v_{i,j} \geq v_{i,j+1}$}\\
        & = v_{i,j+1} \sum_{k=1}^{j+1} y_{i,k} - x_{i,k}.
    \end{align*}
    This implies that Inequality (\ref{equation:prop-sd-utility0}) holds for $j+1$ and completes the induction step, thereby establishing that $u_i(\alloa) \geq u_i(\allob)$.

    \emph{Proving $(2)$.} The proof of this part is very similar to the proof of the previous part. Hence, we will only highlight the point where the proofs differ. Since $\alloa \sdpgt_i \allob$, we know that there exists a $j \in [\ell]$ for which the strict inequality $\sum_{k=1}^j (y_{i,k} - x_{i,k}) > 0$ holds. Following the same inductive argument and additionally using the strict inequality $\sum_{k=1}^j (y_{i,k} - x_{i,k}) > 0$ in the inductive step establishes the strict inequality $u_i(\alloa) > u_i(\allob)$.

\end{proof}

\tightnessPropOne*
\begin{proof}
    Consider the following mechanism for $n=2$ agents: starting from the set of all items $\items$, agent $1$ computes her preference $\prefer_1$ based her valuation and the lexicographic tie-breaking. Let $j^* = \prefer_1(m)$, i.e., $j^*$ is agent $1$'s least favorite item in the ranking. The mechanism assigns $\items\setminus \{j^*\}$ to agent 1 and $\{j^*\}$ to agent 2. It is easy to see that this mechanism is truthful for agent $1$: she only wants to get rid of her least favorite item as per his true utility. And $v_{1, j^*} \leq v_{1,k}$ for all $j \in \items$. Since the mechanism is independent of agent $2$'s reporting, it is also truthful for agent $2$.

    Finally, note that the output allocation is $\mathrm{SD}$ efficient, because we cannot improve an agent's bundle without making the other agent's bundle worse. Specifically, in order to improve the bundle of agent $1$ wrt $\sdgt_1$ we must give item $j^*$ to her, which will make the bundle of agent $2$ strictly worse wrt $\sdgt_2$ since her bundle size would decrease. Similarly, to improve agent $2$'s bundle wrt $\sdgt_2$ we must allocate her another item, which also will make the bundle of agent $1$ worse wrt $\sdgt_1$. This implies that the described mechanism, which is not the serial dictatorship, is truthful and outputs $\mathrm{SD}$ efficient allocations.
\end{proof}

\bivaluedPref*
\begin{proof}
    This proposition essentially follows from the following well-known results from the literature. First, it is known that for binary additive preferences (i.e., $v_{i,j} \in \{0,x\}$ for some fixed $x>0$) there exists a deterministic DSIC mechanism that outputs allocations that are Pareto efficient as well as $\EFone$~\cite{halpern2020fair}. 
    
    Additionally, it is also known that for additive preferences where $0$ is not allowed as a valid report, i.e., $v_{i,j}>0$ (this includes as a special case, the bivalued additive preferences, satisfying $v_{i,j} \in \{a,b\}$ for some $a,b > 0$) there exists DSIC mechanisms --- other than serial dictatorship --- which output Pareto efficient allocations~\cite{amanatidis2017truthful}[Corollary 3.20].
\end{proof}

\section{Missing proofs from Section \ref{subsec:rr is bic}}\label{appendix:rr missing proofs}

\RRsdpoplus*
\begin{proof}
We first show that $\rrpass$ is $\mathrm{SD}^+$efficient. We let $A = (A_1,A_2 \ldots,A_n)$ be an allocation where $A_i$ is the bundle agent $i$ receives. In $\rrpass$, one item is allocated in each timestep. We let $A^t = (A^t_1, A^t_2 \ldots, A^t_n)$ denote the partial allocation of all the allocated items up to and including timestep $t$, where $A^t_i$ denotes the bundle of agent $i$ in $A^t$. We have $A^0_i = \emptyset$ 
and $A^m_i = A_i$ for all agent $i$. Assume, for contradiction that $A$ is not $\mathrm{SD}^+$efficient, i.e., there exists another allocation $B = (B_1, B_2, \ldots, B_n)$ such that 
\[B_j \sdpgeq A_j \text{ for all } j \quad \text{and}\quad B_i \sdpgt A_i \text{ for some }i.\]
We let $B_i^t = B_i \cap A_i^t$, and let $t^* = \min\{t : B_i^t \neq A_i^t \text{ for some }i\}$. Further, let $i^*$ be the agent with the turn to pick in $\rrpass$ at timestep $t^*$ and $\prefer^*$ be the preference order of agent $i^*$, and $\prefer^*(a)$ be the item agent $i^*$ picks at timestep $t^*$ for an index $a \geq 1$. Note that $B^t_{i^*} = A^t_{i^*}$ for all $t < t^*$, since $t^*$ is the first timestep where $A$ and $B$ start to differ.

Define $U \coloneqq \items \setminus \cup_{i=1}^n A_i^{t^*-1}$ (which is same as $\items \setminus \cup_{i=1}^n B_i^{t^*-1}$) to be the set of items that are unallocated after $t^*-1$ timesteps. To arrive at a contradiction, we will first note that the item $\prefer^*(a) \in \items_{i^*}^+$. Otherwise, if $\prefer^*(a) \notin \items_{i^*}^+$, then all the items $j \in U$ must also satisfy $j \notin \cup_{i=1}^n \items_i^+$ since $\rrpass$ will allocate a zero-valued item to an agent $i^*$ only if it is zero-valued by every other agent. However, if all the unallocated items after timestep $t^*-1$ are zero-valued by everyone, and $A_i^{t^*-1} = B_i^{t^*-1}$ for agents $i$, then allocation $B$ cannot stochastically and positively dominate $A$ since each agent's utility will be the same in both the allocations. Therefore, it must be the case that $\prefer^*(a) \in \items_{i^*}^+$, i.e., $a \leq m_{i^*} = |\items_{i^*}^+|$.

To complete the proof, we will now show that $\prefer^*(a) \in \items_{i^*}^+$ implies that $B$ cannot stochastically and positively dominate $A$, contradicting our initial assumption. Recall that, $\rrpass$ constructs the preference order $\prefer^*$ lexicographically, for any $k < a$, we have $B_{i^*, \prefer^*(k)} = A_{i^*,\prefer^*(k)}$. Additionally, since $t^*$ is the first timestep where $A$ and $B$ differ, $A_{i^*, \prefer^*(a)} = 1$ and $B_{i^*, \prefer^*(a)} = 0$ as well. We, therefore, have:
\begin{align*}
    \sum_{\ell = 1}^a B_{i^*, \prefer^*(\ell)} & = \sum_{\ell = 1}^{a-1} B_{i^*, \prefer^*(\ell)}= \sum_{\ell = 1}^{a-1} A_{i^*, \prefer^*(\ell)}\\
    &< \sum_{\ell = 1}^{a-1} A_{i^*, \prefer^*(\ell)} + A_{i^*,\prefer^*(a)} \tag{$A_{i^*,\prefer^*(a)} = 1$}\\
    & =  \sum_{\ell = 1}^{a}A_{i^*, \prefer^*(\ell)}.
\end{align*}
However, note that $\sum_{\ell = 1}^a B_{i^*, \prefer^*(\ell)} < \sum_{\ell = 1}^{a}A_{i^*, \prefer^*(\ell)}$ and the fact that $a \leq m_{i^*}$ contradicts with the assumption that $B$ stochastically and positively dominates $A$. Thus, $A$ must have been $\mathrm{SD}^+$efficient.

For completeness, we will now show that $\rrpass$ is EF1; this argument is the same which is used to show that Round-Robin outputs EF1 allocations. Recall that $\rrpass$ considers agents in a fixed order $\pi$. First, note that agent $1$ in this order does not envy anyone since she always picks first. Similarly, for any $i \in [2,n-1]$, agent $i$ does not envy agents $i+1$ to agent $n$. Now consider agent $i-1$. We will break $\rrpass$ into two stages. In the first stage, each of the first $i-1$ agents receives 1 item, and in the second stage, the rest of $\rrpass$ takes place. We denote the bundle that agent $i$ receives in the second stage of the mechanism as $X_i^2$. Note that agent $i$ is considered first in the second stage of the mechanism, and by the same argument, she does not envy any agent. In particular, for any agent $j \in [i-1]$, we have:
\[v_i(X^2_i) \geq v_i(X_j^2).\]
Note that in the first stage of the mechanism, agent $i$ did not receive any item and agent $j$ received exactly one item. Therefore, we have for any $j \in [i-1]$:
\[v_i(X_i) = v_i(X_i^2) > v_1(X_j^2) \geq v_1(X_j^2 \setminus \{g_j\}),\]
where $g_j$ is the item that agent $j$ receives in the first stage of $\rrpass$.
\end{proof}

\section{Separation between $\PO$, $\sdpoplus$, and $\sdpo$}\label{appendix:efficinecy_examples}
In this section, we will show that the implications established in Lemma \ref{lemma:efficiency_comparison} $(1)$ and $(2)$ are tight, i.e., the other direction of the implications is not true. Towards this, first, we will describe an allocation that is $\mathrm{SD}^+$efficient but not Pareto efficient. Consider an instance with $n=2$ agents and $m=4$ items such that the valuations of the agents are $\val_1 = (5,4,3,2)$ and $\val_2 = (6,1,2,3)$ respectively. Let $\allo = (\allo_1, \allo_2)$ be an allocation wherein agent $1$ receives the first two items and the rest of the items are allocated to agent $2$, i.e., $x_{1,1} = x_{1,2} = x_{2,3} = x_{2,4} = 1$. We will show that $\allo$ is $\mathrm{SD}^+$efficient but not Pareto efficient. The allocation $\allo$ is $\mathrm{SD}^+$efficient because, in order to improve the bundle of agent $1$ wrt $\sdpgt_1$ we must give her another item (she already has her top two items in $\allo$), which will make the bundle of agent $2$ worse wrt $\sdpgt_2$ (since agent $2$ will be left with only one item). Similarly, to improve agent $2$'s bundle wrt $\sdpgt_2$, we must either $(i)$ give her at least three items, or $(ii)$ give her items $1$ and $4$ --- top-2 items as per $\val_2$. Both $(i)$ and $(ii)$ will make the bundle of agent $1$ worse. Thus, $\allo$ is $\mathrm{SD}^+$efficient. Now, to show that $\allo$ is not Pareto efficient consider the allocation $\alloy = (\alloy_1, \alloy_2)$ where agent $2$ gets only item $1$, i.e., $y_{2,1} = y_{1,2} = y_{1,3} = y_{1,4} = 1$. Note that $u_1(\allo) = u_1(\alloy) = 9$, while $u_2(\allo) = 5$ and $u_2(\alloy) = 6$ implying $\alloy$ Pareto dominates $\allo$. Thus, $\allo$ is $\mathrm{SD}^+$efficient but not Pareto efficient.

Finally, we describe an allocation that is $\mathrm{SD}$ efficient but not $\mathrm{SD}^+$efficient. Consider an instance with $n=2$ agents and $m=2$ items such that $\val_1 = (1,1)$ and $\val_2 = (1,0)$; note that $2 \notin \items_2^+$. Let $\allo' = (\allo'_1, \allo'_2)$ be an allocation such that $x'_{1,1} = x'_{2,2} = 1$. This allocation is not $\mathrm{SD}^+$efficient: transferring item $2$ from agent $2$ to agent $1$ doesn't affect agent $2$'s bundle wrt $\sdpgeq_2$, however, it improves agent $1$'s bundle wrt $\sdpgt_1$, i.e., the allocation resulting from this transfer stochastically and positively dominates $\allo$. Contrary to this, the allocation $\allo$ is $\mathrm{SD}$ efficient because in order to improve the bundle of agent $1$ wrt $\sdgt_1$ (or agent $2$ wrt $\sdgt_2$) we must necessarily decrease the bundle size of agent $2$ (agent $1$). This will make the bundle of agent $2$ (agent $1$) strictly worse, since as per the definition of $\sdgeq_i$, a decrease in bundle size necessarily makes the bundle strictly worse. Therefore, $\allo'$ is $\mathrm{SD}$ efficient but not $\mathrm{SD}^+$efficient.

\section{Missing proofs from Section \ref{section:cake-cutting}}\label{appendix:cake cutting proofs}

\IAprop*
\begin{proof}
We will show this via an inductive argument. In particular, we will prove that at the partial allocation $(X_1, X_2, \dots, X_i)$ at the end of iteration $i$ of the outer loop, is proportional for agents $1,2, \dots, i$ for all $i \in \{2, 3, \ldots, n\}$.

\textbf{Base case:} In the iteration $i=2$, agent $1$ splits the interval $[0,1]$ into 2 pieces $C_1^1$ and $C_1^2$ such that $\ft_1(C_1^1) = \ft_1(C^1_2) = 1/2$, where the final equality follows from the fact that $C_1^1 \cup C_1^2 = [0,1]$. Agent $2$ then picks the piece $C^{k^*}_1$ having a higher value wrt $\ft_2$, leaving the remaining piece for agent $1$. The resultant allocation is proportional since agent $1$'s value for both pieces is $1/2$, and agent $2$ gets to pick her favorite piece, implying that  $\ft_2(C^{k^*}_1) \geq \ft_2([0,1])/2 = 1/2$.

\textbf{Induction step:} For $i>2$, let $\alloi' = (X'_i, X'_2 \dots, X'_{i-1})$ denote the allocation before the arrival of agent $i$ (i.e., at the beginning of iteration $i$). Assuming that $\alloi'$ is proportional for agents $[i-1]$, we will now show that the allocation $\alloi = (X_1, X_2, \ldots, X_i)$ at the end of iteration $i$ will be proportional for agents $[i]$. First, we will show that $\alloi$ is proportional for agents $j \in [i-1]$. Towards this, note that $\alloi'$ is proportional for agents $j \in [i-1]$, i.e., for any agent $j \in [i-1]$, we have, $
\ft_j(X'_j) \geq \frac{1}{i-1}$. Additionally, from the definition of the Split-Equal subroutine, we know that  $\ft_j(C_j^k) = \frac{1}{i}\ft_j(X'_j)$  for all $k \in [i]$ and $j \in [i-1]$, where pieces $C_j^1, C_j^2, \ldots, C_j^i$ are obtained from the performing Split-Equal operation on inputs $(\ft_j, X'_j, i)$. On combining these two observations, we get that for all $j \in [i-1]$,

\[\ft_j(X_j) =  \ft_j(X'_j \setminus C_j^{k^*}) = \left(1-\frac{1}{i}\right)\ft_j(X_j) \geq \frac{i-1}{i}\cdot \frac{1}{i-1} = \frac{1}{i},\]


where $C_j^{k^*}$ is the piece picked by agent $i$ from agent $j$ in iteration $i$. This implies that $\alloi$ is proportional for agents $j \in [i-1]$. Finally, we consider agent $i$. Note that by definition of the mechanism, agent $i$ picks her favorite interval from each agent $j \in [i-1]$, this implies that,
\begin{align*}
    \ft_i(X_i) = \sum_{j = 1}^{i-1} \max_k \ft_i(C_j^k) & \geq \sum_{j = 1}^{i-1} \frac{1}{i}\sum_{k = 1}^i\ft_i(C_j^k) \tag{maximum is at least the average}\\
    & = \sum_{j = 1}^{i-1} \frac{1}{i} \ft_i(X'_j) \tag{$X'_j = \cup_{k=1}^{i} C_j^k$}\\
    & = \frac{1}{i} \ft_i([0,1]) = \frac{1}{i},
\end{align*}
where the penultimate inequality follows from the fact that $\cup_{j=1}^{i-1} X'_j = [0,1]$, i.e., $\alloi'$ is a complete allocation among agents $[i-1]$. This implies that $\alloi$ is proportional for agent $i$ as well. This completes the induction step and concludes the proof.
\end{proof}

\end{document}